\documentclass[a4paper,12pt]{article}
\usepackage{mathtools}
\usepackage{amsfonts}
\usepackage{amsthm}
\usepackage{amsbsy}
\usepackage{enumerate}
\usepackage{bbm}
\usepackage{fullpage}
\usepackage{color}
\usepackage{microtype}
\usepackage{hyperref}
\usepackage{graphbox}
\usepackage{todonotes}
\usepackage{float}
\usepackage{comment}

\usepackage{hyperref}

\theoremstyle{plain}
\newtheorem{thm}{Theorem}[section]

\theoremstyle{definition}

\newcommand{\bbR}{\mathbb{R}}

\newcommand{\cP}{\mathcal{P}}
\newcommand{\cS}{\mathcal{S}}

\newcommand{\supp}{\mathrm{supp}}

\begin{document}

\title{The effect of timescale separation on the tipping window for chaotically forced systems}

\author{Raphael R\"omer and Peter Ashwin\\
Department of Mathematics and Statistics,\\ University of Exeter, Exeter EX4 4QF, UK}

\maketitle
Author Accepted Manuscript. Accepted for publication in SIAM Journal on Applied Dynamical Systems (SIADS)

\begin{abstract}
Tipping behavior can occur when an equilibrium of a dynamical system loses stability in response to a slowly varying parameter crossing a bifurcation threshold, or where noise drives a system from one attractor to another, or some combination of these effects. Similar behavior can be expected when a multistable system is forced by a chaotic deterministic system rather than by noise. In this context, the chaotic tipping window was recently introduced and investigated for discrete-time dynamics. 
In this paper, we find tipping windows for continuous-time nonlinear systems forced by chaos. We characterize the tipping window in terms of forcing by unstable periodic orbits of the chaos, and we show how the location and structure of this window depend on the relative timescales between the forcing and the responding system. We illustrate this by finding tipping windows for two examples of coupled bistable ODEs forced with chaos. Additionally, we describe the dynamic tipping window in the setting of a changing system parameter.
\end{abstract}

\tableofcontents

\section{Introduction} \label{sec:Intro}

In many models of real-world systems, different processes can evolve over very different timescales. If one subsystem evolves chaotically on a rapid timescale and forces another slower subsystem, the faster subsystem can be seen as similar to a noise process driving the slower part. Conversely, if the forcing subsystem evolves much more slowly than the system it forces, then the forcing system can be seen as a slowly varying parameter of the forced subsystem.

Such timescale separations allow for substantial simplifications of models for processes in a variety of applications from ecology, neuroscience, and chemistry to economics and data science; see for example \cite{berglund2006noise,kuehn2015multiple}. It is also a common assumption in climate modeling, made explicit by Hasselmann \cite{Hasselmann1993optimal, Lucarini2023theoretical}, that weather processes are rapid and chaotic and that variability of climate variables (such as global mean surface temperature) occurs on a slower timescale driven by the weather fluctuations. This suggests that climate can be modeled using stochastic differential equations, where the chaotic driving is approximated by a white noise (Wiener process) driving. 

Recent progress has shown that this is indeed the case for many systems in an appropriate limit \cite{Kelly:2017, gottwald2017, Melbourne_2011}. In particular, if the fast dynamics are ergodic and exponentially rapidly mixing, homogenization theory \cite{Kelly:2017, Melbourne_2011} gives a setting where this approximation is correct in the limit of the chaotic forcing subsystem being infinitely fast compared to the forced subsystem. Away from this limit, modeling fast dynamics as a noise process can clearly lead to even qualitatively wrong results; for example, if the chaotic forcing is bounded, there are some transitions between attractors that are possible in stochastic differential equations but not for systems forced with bounded dynamics.

This qualitative difference between chaotic and stochastic forcing can be seen when considering a bistable system forced by either a Wiener process or bounded chaos. Initialized close to one attractor forcing with a Wiener process will almost surely lead to tipping to the basin of attraction of the other attractor in finite time, and the expected tipping time's asymptotics for small forcing strengths are given by the Eyring-Kramers formula \cite{berglund2006noise, Boerner2024}.
In contrast, bounded chaotic forcing does not necessarily lead to tipping at all. Even if the forced bistable system is close to a bifurcation, tipping will strongly depend on the realization of the chaotic forcing; some forcing trajectories might lead to tipping of the bistable system, but others might not.

This effect has been characterized as the appearance of a chaotic tipping window \cite{AshwinNewmanRoemer:2024} near tipping points of systems that are chaotically forced. The tipping window is the region in the parameter space of a forced system for which it tips for some specific trajectory of a chaotic forcing but not for others. Outside the chaotic tipping window, the system will either never tip or always tip independently of the specific forcing trajectory on the bounded chaotic attractor of the forcing system.
Note that the entrance into the chaotic tipping window can also be seen from the perspective of the full system as a loss of stability of an attractor as a result of a crisis.
The detailed example in \cite{AshwinNewmanRoemer:2024} considered discrete time and a particularly simple chaotic forcing such that many results could be rigorously proven.
In this paper, we consider the case of continuous-time ordinary differential equations (ODEs) and analyze how varying the relative timescales between the forcing and forced subsystem affects the chaotic tipping windows. 

In Section~\ref{sec:TippingWindow}, we first generalize the notion of the tipping window for an additively chaotically forced system independent of relative timescales. Then, we derive analytical results in limiting cases of infinite timescale-separation in Section \ref{sec:inf_ts_sep} for a one-dimensional additively forced system; in particular, we show that the extrema of the forcing are important to determine tipping in the case of slow forcing, while the mean of the forcing is important to determine tipping for the case of rapid periodic forcing.

In Section \ref{sec:Example}, we verify these results for two additively, chaotically forced bistable ODE examples. The first system consists of dynamics in a double-well potential, with its bifurcation parameter additively forced by the first variable of the Lorenz-63 system in the chaotic regime. The second is a Stommel system similarly forced by Lorenz-63. We compute the tipping window of these systems when changing the relative speed between forcing and response by examining the response of the bistable systems to unstable periodic orbits (UPOs) on the chaotic attractor of the forcing system. We verify that in the limiting cases of timescale separation, indeed, the mean and extrema of the forcing UPOs allow us to determine the tipping window. For cases of intermediate relative timescales, there is a nontrivial rearrangement of the UPOs that determines the boundaries of the tipping window. We examine this for a specific example and find that there is a balance between the relative timescales and the duration that the UPOs can overshoot a bifurcation point.

In Section~\ref{sec:DynamicTippingWindow}, we turn to the {\it dynamic tipping window} also introduced in \cite{AshwinNewmanRoemer:2024}. This is a non-autonomous version of the tipping window, as the bifurcation parameter is now ramped at a finite rate. In the limit of infinitely slow parameter ramping, it limits to the tipping window of the autonomous case. We end with a brief discussion in Section~\ref{sec:discuss}.

\subsection{Tipping points and chaotic forcing timescales}
\label{sec:ChaoticForcingAndTippingPoints}

A {\em tipping event} occurs when a nonlinear system has a large or rapid change in state in response to small or slow changes in forcing. These events can happen in nonlinear systems as a result of various mechanisms \cite{Ashwinetal:2012,Ashwinetal:2017}. \emph{Bifurcation tipping} (B-tipping) occurs when a system parameter slowly crosses a bifurcation point that results in the loss of any nearby attractor. \emph{Rate-induced tipping} (R-tipping) can occur when parameter changes are fast enough that the adiabatic approximation breaks down.  \emph{Noise-induced tipping} (N-tipping) is another tipping mechanism that is usually studied in systems driven by Gaussian noise (i.e. a Wiener process), where the unboundedness of the noise means that in the presence of two or more attractors, there will almost surely be tipping between them on waiting a sufficiently long time. 
Non-trivial dynamics on the attractor may bring new effects such as \emph{partial tipping} \cite{Alkhayuon:2018,ashwin2021physical}, \emph{phase tipping} \cite{alkhayuon2021phase}. In the presence of chaotic forcing, the \emph{chaotic} and \emph{dynamic tipping windows} \cite{AshwinNewmanRoemer:2024} can appear.

Let us consider a forced nonlinear system in the timescale $t$ of the response, inspired by \cite{AshwinNewmanRoemer:2024}, where the forcing and response timescales have a ratio $\gamma>0$, namely
\begin{equation}
	\begin{aligned}  \label{eq:fullSyst}
		\frac{dx}{dt} = & f(x,\eta(t), a \phi(y)) \\
		\frac{dy}{dt} = &\gamma g(y) 
	\end{aligned}
\end{equation}
with $x\in\bbR^d$, $y\in\bbR^n$, $\eta:\bbR\rightarrow \bbR$, $\phi:\bbR^n\rightarrow \bbR$, $g:\bbR^n\rightarrow \bbR^n$, $f:\bbR^d \times \bbR\times\bbR\rightarrow \bbR^d$, and $a\geq 0$. We assume that $f$, $\eta$, $\phi$, and $g$ are smooth functions of their arguments. We assume that $f$ has B-tipping points of equilibria for the case $a=0$ and slow variation of $\eta$. The parameter $a$ scales the amplitude of the forcing $\phi(y)$, which is assumed to be bounded. We assume that the $y$ dynamics has a physical ergodic invariant measure $m$ (i.e., a positive Lebesgue measure set of initial conditions has empirical measures that converge to $m$ \cite{Young:2016}). We suppose that the initial condition $y(0)$ is typical with respect to this physical measure. Often, we consider {\em chaotic forcing}, i.e., the physical measure $m$ is supported on a chaotic attractor. If we write the timescale of the forcing as $\tau= \gamma t$, then we can express (\ref{eq:fullSyst}) as
\begin{equation}
    \begin{aligned}  \label{eq:fullSysttau}
    \frac{dx}{d\tau} = &\ \gamma^{-1}f(x,\eta(\gamma^{-1} \tau),a \phi(y)) \\
    \frac{dy}{d\tau} = &\ g(y). 
\end{aligned}
\end{equation}
The function $\eta(t)$ describes the value of the system parameter for any time $t$. If $\eta(t)$ is assumed to change on a time scale much slower than either $x$ or $y$, then the dynamics of (\ref{eq:fullSyst}) can be largely understood in terms of the following {\em frozen system with chaotic forcing} (in terms of the response timescale) by fixing $\eta$:
\begin{equation}
    \begin{aligned}  \label{eq:frozenSyst}
    \frac{dx}{dt} = & f(x,\eta, a \phi(y)) \\
    \frac{dy}{dt} = &\gamma g(y). 
\end{aligned}
\end{equation}
We restrict (\ref{eq:frozenSyst}) to a special case of an additively forced equation for $x\in\bbR$ given by
\begin{equation}
    \begin{aligned}  \label{eq:frozenSystf0}
    \frac{dx}{dt} = & f(x) + \eta + a \phi(y) \\
    \frac{dy}{dt} = &\gamma g(y). 
\end{aligned}
\end{equation}
where $\frac{dx}{dt}=f(x)+\eta$ is bistable for some values of $\eta$. The system (\ref{eq:frozenSystf0}) is equivalent to
\begin{align}
\frac{dx}{dt}&=f(x)+\eta+a\phi(y(\gamma t)) \label{eq:addforcedSyst}\\
     \frac{dy}{d\tau}&=g(y). \label{eq:addchaoticSyst}
\end{align}

We discuss in the next Section how a bifurcation for (\ref{eq:frozenSystf0}) with $a=0$ can become a tipping window for $a>0$ and how $\gamma$ can influence the form of this tipping window.

\section{The chaotic tipping window}
\label{sec:TippingWindow}

For some forcing amplitude $a>0$, the tipping window $W(a)$ characterizes the set of parameters $\eta$ for which the system (\ref{eq:frozenSystf0}) undergoes tipping events for some realizations of the forcing but not for others. This concept was introduced in \cite{AshwinNewmanRoemer:2024} and investigated there in the context of discrete dynamical systems. 
Now, suppose that the unforced equation (often called the {\em response system})
\begin{align} \label{eq:addunforcedSyst}
    \frac{dx}{dt} = &\ f(x)+\eta
\end{align}
has a bifurcation point $\eta^*$ that can cause B-tipping. Specifically, suppose there are two attractors for some range $\eta^{\dag}<\eta<\eta^*$ and only one attractor for $\eta>\eta^*$, and that one of the attractors disappears at a saddle-node bifurcation $\eta=\eta^*$.

Consider a small but positive forcing strength $a>0$ for (\ref{eq:addforcedSyst}). In this context, the \emph{chaotic tipping window} about $\eta^*$ will be an interval $W(a)$ such that whenever $\eta \notin W(a)$, the Equation (\ref{eq:addforcedSyst}) has either one, or two attractors independent of the chosen forcing solution $y(\tau)$ of (\ref{eq:addchaoticSyst}), and such that
$$
\lim_{a\rightarrow 0} \sup\{d(x,\eta^*)~:~x\in W(a)\}=0,
$$
i.e. such that both endpoints of the interval $W(a)$ limit to $\eta^*$ for $a\rightarrow0$. Note that if the $y$-dynamics are not chaotic, but $y$ is a critical point of Equation \eqref{eq:addchaoticSyst}, the tipping window is a degenerate interval, i.e., a point.

Within the chaotic tipping window, i.e., for $\eta \in W(a)$, the dynamics of the full system (\ref{eq:addforcedSyst} \& \ref{eq:addchaoticSyst}) will depend on the specific trajectory of the chaotic forcing (\ref{eq:addchaoticSyst}). For some forcing trajectories of the chaotic system (\ref{eq:addchaoticSyst}), there is only one attractor for (\ref{eq:addforcedSyst}): the chaotically forced system tips earlier than the unforced equation (\ref{eq:addunforcedSyst}) would have tipped as a result of adiabatically increasing $\eta$.  
Other trajectories of the chaotic system (\ref{eq:addchaoticSyst}) will give more than one attractor for the equation (\ref{eq:addforcedSyst}): the chaotically forced system tips later than the unforced equation (\ref{eq:addunforcedSyst}) would have tipped as a result of adiabatically increasing $\eta$. 

The chaotic tipping window was introduced in \cite{AshwinNewmanRoemer:2024}, and we follow this approach closely. However, here we slightly modify the previous definition and define the {\em set of tipping windows} $W(a)$ as an intersection of two closed sets. As $W(a)$ is typically just a single interval, i.e. the set of tipping windows has just a single element, we use the terms {\em tipping window} and the {\em set of tipping windows} interchangeably here, but always mean the definition of $W(a)$ below. 

Consider the set $\cS(m)$ of all ergodic invariant measures for the $y$-dynamics that are supported within $\supp(m)$, the support of the physical ergodic invariant measure $m$. We write $\cP(m)$ to denote the unstable periodic orbits contained within $\supp(m)$; clearly (if we identify a periodic orbit with an ergodic measure supported on it) $\cP(m)\subset \cS(m)$. Moreover, for many chaotic systems, $\cP(m)$ is dense within $\cS(m)$ in the C* topology (see e.g. \cite{sigmund1972space,yuan1999}).  We introduce the notion that the system (\ref{eq:frozenSystf0}) has one attractor when forced by $s\in\cS(m)$ if there is a full measure set of trajectories $y$ w.r.t. $s$ such that (\ref{eq:frozenSystf0}) has only one attractor for $x$. Inspired by \cite{AshwinNewmanRoemer:2024}, for fixed parameters we define the set
$$
M(s,a):=\{\eta~:~ \mbox{(\ref{eq:frozenSystf0}) has one attractor for given $a$ when forced by } s \in \cS(m)\},
$$
and the set of chaotic tipping windows as
$$
W(a):=\{\eta~:~\eta\in \overline{M(s,a)}\cap \overline{M^c(s',a)}~\mbox{ for some }s,s'\in \cS(m)\}.
$$
As \cite{AshwinNewmanRoemer:2024} notes, the elements $W(a)$, i.e. the single tipping windows, correspond to those parameter values where different forcing trajectories may give a different number of attractors.
In the case where $M(s,a)$ is a semi-infinite interval $(\eta_c(s,a),\infty)$ for all $s\in\cS(m)$ the chaotic tipping window will be an interval $W(a)=[\eta_-(a),\eta_+(a)]$. The endpoints
\begin{equation}
\begin{aligned}
\eta_-(a)=&\inf \{ \eta_c(s,a)~:~s\in\cS(m)\},\\\eta_+(a)=&\sup \{ \eta_c(s,a)~:~s\in\cS(m)\}
\end{aligned}
\label{eq:eta+-}
\end{equation}
of this window can be thought of as solving a problem of ergodic optimization. Measures realizing these extremes are often found to be periodic \cite{jenkinson2019ergodic}; this was the case in \cite{AshwinNewmanRoemer:2024}, but here we find that this is not always the case. In \cite{AshwinNewmanRoemer:2024}, it was found that $\eta_c(s,a)$ is not simply an averaged observable but the boundary of a region of multistability, and it is not obvious how to generalize these results. The threshold on entry to the chaotic tipping window can also be understood in terms of a \emph{non-autonomous saddle-node bifurcation}~\cite{Anagnostopoulou:2012} of the non-autonomous system (\ref{eq:frozenSystf0}).

\subsection{Tipping windows at the limits of timescale separation}
\label{sec:inf_ts_sep}

In this section, we argue that for limiting timescale separation $\gamma\rightarrow0$, only the maximum, and for $\gamma\rightarrow\infty$, only the mean of a given forcing trajectory are important for understanding the chaotic tipping window.

For definiteness, we consider the case where Equation \eqref{eq:addunforcedSyst} only has equilibrium attractors and a hysteresis loop on increasing $\eta$. More precisely, we assume that there is bistability for $\eta^\dag<\eta<\eta^*$ with saddle-node bifurcations at $(\eta,X)=(\eta^*,X^*)$ and $(\eta^\dag,X^\dag)$. We assume that for $\eta\notin \{\eta^\dag,\eta^*\}$ all equilibria are hyperbolic, and that there is a branch of attractors $X(\eta)$ such that $\lim_{\eta\nearrow\eta^*}X(\eta)=X^*$.

We turn to cases with \ $a$ \ small and \ $a>0$ in \eqref{eq:frozenSystf0}. Recall that $\tau$ is the timescale of the forcing $y(\tau)$, $t$ is the timescale of the response $x(t)$, and $\tau = \gamma t$. 
For any trajectory $y(\tau)$, we write
$$
\widetilde{\Phi}(\tau):=\phi(y(\tau)),
$$
and define
$$
\overline{\Phi}:=\lim_{T\rightarrow \infty} \frac{1}{T}\int_{\tau=0}^T \phi(y(\tau))\, d\tau
$$
(so that $\widetilde{\Phi}(\tau)=\overline{\Phi}+\Phi(\tau)$ with $\Phi(\tau)$ having zero mean), and define 
$$
\Phi^+:=\sup_{\tau>0} \phi(y(\tau)).
$$
We say $y$ is {\em generic} for $s\in\cS(m)$ if
$$
\begin{aligned}
&\overline{\Phi}=\overline{\Phi}(s):=\int_{y\in\supp(s)} \phi(y)\, ds(y) ~\mbox{ and}\\
&\Phi^+=\Phi^+(s):=\sup_{y\in\supp(s)} \phi(y).
\end{aligned}
$$

For small $\gamma$, the forcing in the system \eqref{eq:frozenSystf0} evolves slowly compared to the unforced equation (or response system) (\ref{eq:addunforcedSyst}), and thus, the latter reacts adiabatically to the forcing.
In the limit $\gamma \rightarrow 0$, we show that one of the response system’s attractors loses stability if $\eta+a \phi(y(\tau))$ exceeds $\eta^*$ at any time, i.e. if $y$ is generic for $s\in\cS(m)$ then the critical $\eta$ is at:
$$
\eta_c(s,a)=\eta^*- a \Phi^+(s).
$$
Hence, lower and upper bounds of the chaotic tipping window (\ref{eq:eta+-}) in the limit $\gamma\rightarrow 0$ can be given by
\begin{equation}
\begin{aligned} \label{eq:limsTipWin_gamma0}
    \eta_-(a)=&\eta^*-a\sup\left\{ \Phi^+(s) ~:~s\in\cS(m)\right\},\\
    \eta_+(a)=&\eta^*-a\inf\left\{ \Phi^+(s)~:~s\in\cS(m)\right\}.
\end{aligned}
\end{equation}
In the opposite case, where $\gamma$ is large, the forcing evolves rapidly compared to the system of interest, and the latter only ``sees'' the running mean value of the forcing. For UPO forcing, i.e. $p \in \cP(m)$, in the limit $\gamma \rightarrow \infty$, one of the response system’s attractors loses stability if $\eta+a\phi(y(\tau))$ exceeds $\eta^*$ on average.
In the case $\gamma \rightarrow \infty$, we show that, if $\cP(m)$ is dense in $\cS(m)$, the lower and upper bounds of the chaotic tipping window can be given by the trajectories $p \in \cP(m)$ whose mean values are maximal or minimal:
\begin{equation}
\begin{aligned} \label{eq:limsTipWin_gammaInf}
    \eta_-(a)=&\eta^*-a\sup \left\{\overline{\Phi}(p)~:~ p\in\cP(m)\right\},~~~\\
    \eta_+(a)=&\eta^*-a\inf \left\{\overline{\Phi}(p)~:~ p\in\cP(m)\right\}.
\end{aligned}
\end{equation}
For intermediate timescale ratio $\gamma$, deriving analytical results for the location of the attractor crisis seems highly non-trivial, hence we explore this numerically in Section~\ref{sec:Example}. 

\begin{thm}\label{thm:limits}
    Suppose that (\ref{eq:addforcedSyst}) with $a=0$ 
    has bistability for $\eta^\dag<\eta<\eta^*$ with saddle-node bifurcations at $(\eta,X)=(\eta^*,X^*)$ and $(\eta^\dag,X^\dag)$. Assume that all limit points away from these saddle-node bifurcation points are hyperbolic equilibria. Suppose that $m$ is a physical measure and $y(\tau)$ is a solution of (\ref{eq:addchaoticSyst}) contained within $\supp(m)$. 
	\begin{itemize}
        \item[(a)] For sufficiently small $a>0$ and $y$ generic for $s\in\cS(m)$, in the limit $\gamma\rightarrow 0$ one of the attractors of the response system, Equation(\ref{eq:addforcedSyst}), loses stability at
        $$
        \eta_c=\eta^*-a \Phi^+(s).
        $$
        In particular, (\ref{eq:limsTipWin_gamma0}) gives the limits of a chaotic tipping window in the limit $\gamma\rightarrow 0$.
        \item[(b)] For sufficiently small $a>0$ and $y$ generic for $p\in\cP(m)$, in the limit $\gamma \rightarrow \infty$ one of the attractors of the response system, Equation (\ref{eq:addforcedSyst}), loses stability at
        $$
        \eta_c=\eta^*-a \overline{\Phi}(p).
        $$
        If periodic measures $\cP(m)$ are dense in $\cS(m)$, then \eqref{eq:limsTipWin_gammaInf}  gives the limits of a chaotic tipping window for $\gamma\rightarrow\infty$. 
    \end{itemize}        
 \end{thm}

\begin{proof}
Recall that $x\in\bbR$ and note that Equation (\ref{eq:addforcedSyst}) can be written as
\begin{equation}\label{eq:proofEq1}
    \begin{aligned}
    \frac{dx}{dt} = & f(x)+\eta + a \widetilde{\Phi}\left(\gamma t \right).
\end{aligned}
\end{equation}
For (a), we can view Equation \eqref{eq:proofEq1} as the first component of a two timescale system
\begin{equation} \label{eq:proofEq2}
    \begin{aligned}
    \frac{dx}{dt} = & f(x)+\eta + a \widetilde{\Phi}\left(s \right)\\
    \frac{ds}{dt} = & \gamma
\end{aligned}
\end{equation}
which, in the singular limit $\gamma\rightarrow 0$ will have a critical set $\eta+a\widetilde{\Phi}(s)=-f(x)$ with fold at $\eta+a\widetilde{\Phi}(s)=\eta^*$. If $\eta^{\dag}<\eta+a\widetilde{\Phi}(s)<\eta^*$ for all $s$ (which is the case for some choices of $\eta$ if $a$ is small enough), then the first equation of the system \eqref{eq:proofEq2} (i.e. the fast subsystem) will be linearly stable. At $\eta+a\Phi^+=\eta^*$, the fast subsystem loses stability, hence the first part of the stated result. The greatest and least values of $\eta_c$ are then given by considering all possible invariant measures $s\in\cS(m)$, hence the second part of the stated result.

For (b), we use a corollary building on the Averaging Theorem as stated in \cite[Theorem 4.4.1]{GuckenheimerHolmes1983} (see also \cite[Chapter 9.6]{kuehn2015multiple}). With $\tau=\gamma t$, Equation \eqref{eq:addforcedSyst} can be written as
\begin{align}
	\frac{dx}{d\tau} = &\ \gamma^{-1}\left(f(x)+\eta+a \phi(y(\tau))\right).
\end{align}
For a periodic $y(\tau)$ with period $T$ that is generic for $p\in\cP(m)$ we can write the forcing $\overline{\phi}(y(\tau))$ as mean $\overline{\Phi}$ plus a mean-zero period-$T$ oscillation $\Phi(\tau)$ around $\overline{\Phi}$. We write
\begin{align}
	\frac{dx}{d\tau} = &\ \gamma^{-1}\left(f(x)+\eta+ a \overline{\Phi} + a\Phi(\tau)\right).
\end{align}
Writing $\eta_a:=\eta+a\overline{\Phi}$ we have
\begin{align} \label{eq:AverTheorNonAutoEq}
	\frac{dx}{d\tau} = &\ \gamma^{-1}\left(f(x)+\eta_a + a\Phi(\tau)\right).
\end{align}
We define the associated autonomous averaged system as
\begin{align} \label{eq:AssAverAutoSys}
	\frac{dz}{d\tau} = &\ \gamma^{-1}\frac{1}{T}\int_0^T f(z)+\eta_a + a\Phi(\tau)\ d\tau = \gamma^{-1}\left(f(z)+\eta_a\right).
\end{align}
We write $\epsilon = 1/\gamma$ and note that for 
$\eta\not \in \{\eta^\dag,\eta^*\}$ all limit sets are hyperbolic equilibria with transverse intersections of invariant manifolds (this is automatic for $x\in\bbR$) so by \cite[Theorem 4.4.1]{GuckenheimerHolmes1983}, if $\gamma$ is large enough then the Poincare map of \eqref{eq:AverTheorNonAutoEq} and \eqref{eq:AssAverAutoSys} are topologically equivalent. In particular, the number of attractors of the two systems are the same for large enough $\gamma$. 
Hence at $\eta_a=\eta_*$, for small $a$ and in the limit $\gamma\rightarrow \infty$, the number of attractors of \eqref{eq:AverTheorNonAutoEq} changes from two to one on passing through $\eta_c=\eta^*-a \overline{\Phi}$. Hence the first part of the result (b). The second part follows from the assumed density of $\cP(m)$ in $\cS(m)$.
\end{proof}

\section{Tipping windows for bistable ODEs forced with chaos}
\label{sec:Example}

In the previous section, we gave some results for a one-dimensional multistable equation forced by a chaotic system in the limit of infinite timescale separation. We now consider two examples to verify the results above and to extend them to cases of an arbitrary timescale ratio.

\subsection{Example I: 1D-double-well system forced by Lorenz-63}
\label{sec:DoubleWellLorenz}

We consider dynamics in a one-dimensional double-well potential driven by trajectories on the Lorenz-63 attractor. The system is of the form (\ref{eq:frozenSystf0}) with $x\in\bbR$ and $y\in \bbR^3$,

\begin{equation} \label{eq:LorDoubleWell}
\left.\begin{aligned}
    \frac{d}{dt}{x} &= f(x)+ \eta + a \phi(y)  \\
    \frac{d}{dt}{y} &= \gamma g(y)
    \end{aligned}\right\}
\end{equation}
with $f(x)$ given by the dynamics in a double well potential
\begin{equation} \label{eq:doubleWell}
    f(x)= 3x-x^3,
\end{equation}
and $\phi(y)=y_1$ defined as the projection onto the first coordinate of the $y$-dynamics which are given by the Lorenz-63 system: 
\begin{equation}\label{eq:LorenzPartDoubleWell}
\left.\begin{aligned}
    g_1(y) &:= \sigma (y_2 - y_1) \\
    g_2(y) &:=y_1(\rho-y_3)-y_2 \\
    g_3(y) & := y_1 y_2 - \beta y_3.
\end{aligned}\right\}.
\end{equation}
For $\eta=0$, the unforced $x$-dynamics ($a=0$) has two stable equilibria in $x_{\pm}=\pm\sqrt{3}$, and an unstable equilibrium at $x_s=0$.

More generally, for $\eta\neq \pm 2$, the number and stability of the equilibria are as follows. For $\eta<-2$ there is one stable equilibrium $x_-(\eta)<-2$. For $-2<\eta<2$ there are two stable equilibria $x_+(\eta)\in(1,2)$ and $x_-(\eta)\in(-2,-1)$ separated by an unstable equilibrium $-1<x_s(\eta) <1$. For $\eta>2$ there is one stable equilibrium $x_+(\eta)>2$.
At $\eta^*=\pm 2$ there are saddle-node (fold) bifurcations where $x_{+}(\eta)$ or $x_{-}(\eta)$ meets $x_s(\eta)$.
Note that the location of the equilibria depends on $\eta$.

Turning to cases with $a>0$, the $y$-dynamics become important. Initializing the $x$-dynamics in $x_0<x_-(\eta)$ (i.e. in the basin of the attractor $x_-(\eta)$), setting $0<a\ll1$, and choosing different initial conditions of the $y$-dynamics close to the Lorenz-63 attractor, can show tipping to $x_+(\eta)$ depending on the exact initial condition of the forcing $y_0$, the value of $\eta$, and the value of $\gamma$. 

\begin{figure}
    \centering
    \includegraphics[align=t,width=12cm]{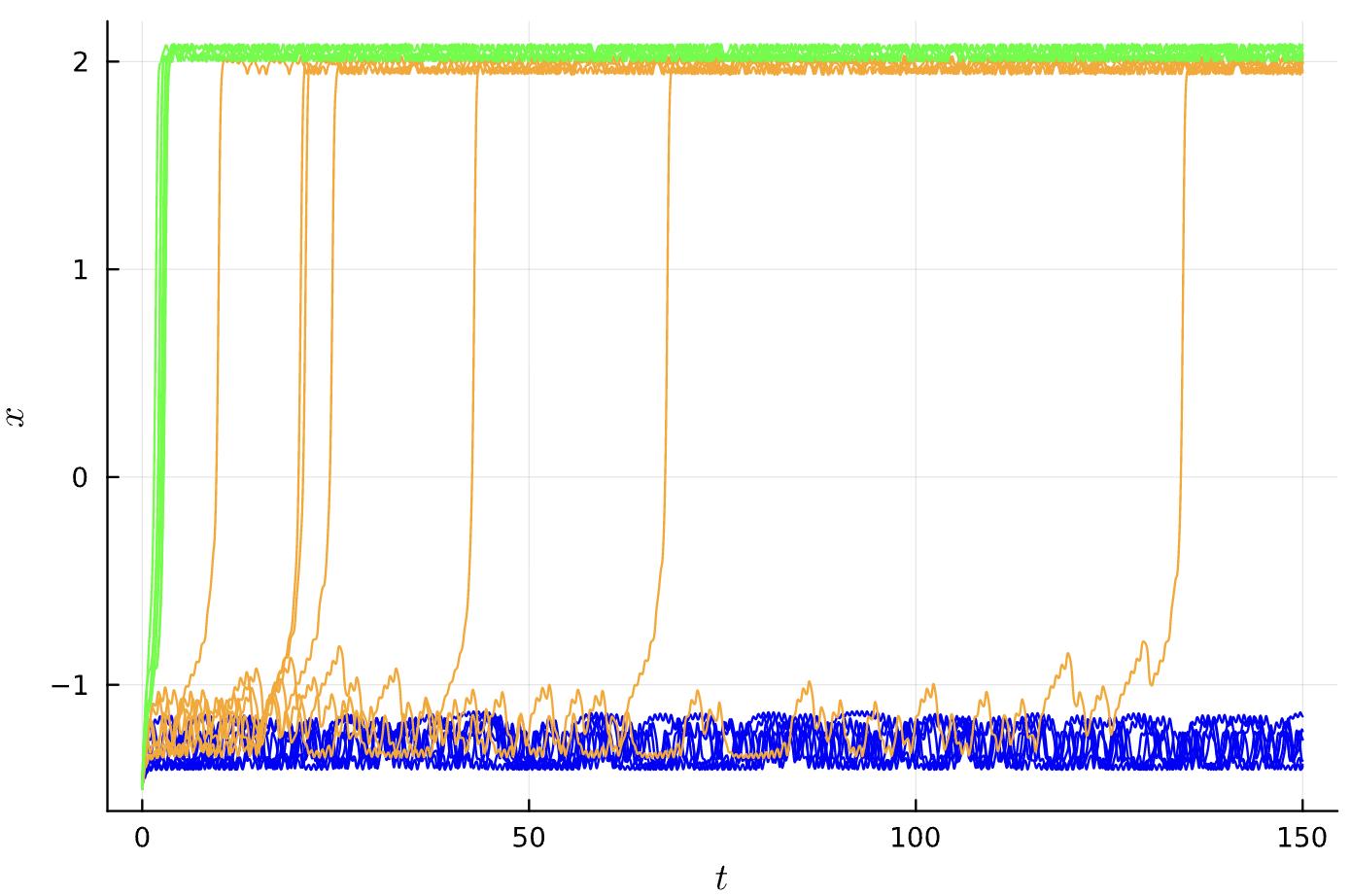}  
    \caption{Typical trajectories of the Lorenz-double-well system \eqref{eq:LorDoubleWell} with $a=0.03$, $\gamma=1$, randomly chosen initial conditions of the Lorenz-$y$-dynamics (i.e. typical with respect to the physical ergodic measure $m$ on the Lorenz-63 attractor), and three different values of $\eta$. The $x$-dynamics are initialized at $x=-1.5$ in the basin of attraction of the double-well attractor $x_-(\eta)$. The blue curves ($\eta=1.7$) remain in the vicinity of the attractor of the $x$-dynamics they started in. The orange curves ($\eta=1.85$) all leave the vicinity where they were initialized in (``they tip''), but the time of leaving this vicinity strongly depends on the specific $y$-trajectory that forces the $x$-dynamics. The green curves ($\eta=2.01$) always immediately leave towards the attractor $x_+(\eta)$, independently of the specific Lorenz trajectory. The $y$-trajectories are generated by taking the final state of the previous ensemble member as the initial condition for the next.}
    \label{fig:DW_typical}
\end{figure}

Figure \ref{fig:DW_typical} shows some typical trajectories for three different choices of $\eta$. They are ``typical'' in the sense that the initial conditions of the $y$ dynamics are chosen typically with respect to the natural ergodic measure $m$ on the chaotic attractor of the $y$ dynamics. The behavior of the orange trajectories suggests that $\eta=1.85$ lies within the tipping window defined in Section \ref{sec:TippingWindow}, since some trajectories immediately go to $x_+(\eta)$ (which we refer to as: ``they tip'') and others only much later within the integration time, depending on the exact trajectory of the $y$ dynamics. However, eventually all trajectories would go to $x_+(\eta)$ for one of the following reasons: (1) with probability 1, the randomly chosen $y$ initial conditions on the forcing attractor are not in the set of initial conditions for which the resulting $x$ trajectories would stay in the vicinity of the attractor $x_-(\eta)$, and upon waiting long enough the forcing trajectory would shadow an extreme forcing trajectory for long enough such that the system would tip. Or (2) with probability 0, the randomly chosen $y$ initial conditions on the forcing attractor are in the set of trajectories for which the resulting $x$ trajectories would never leave the vicinity of the attractor $x_-(\eta)$, but in this case any numerical error will, at some point, make the trajectory leave the possibly still existing measure zero set of remaining trajectories, so eventually the system would tip also in this case. 

All green trajectories go to $x_+(\eta)$ immediately, and no blue trajectories go to $x_+(\eta)$ within the given time, suggesting that the associated $\eta$-values are not in the tipping window. Computing the exact location and size of the tipping window is, however, far from trivial, as it involves exploring a large set of atypical trajectories on the Lorenz-63 attractor that yield the most extreme forcing behavior.

\subsubsection{Unstable Periodic Orbits}
A typical chaotic attractor will support an infinite number of unstable periodic orbits (UPOs) \cite{so1996detecting,bradley2002recurrence}, and the set of periodic orbits will be dense on the attractor \cite{sigmund1972space,yuan1999}. Thus, understanding the response of a system forced by UPOs on an attractor will help to understand the dynamics of that system forced by more general trajectories on the same attractor. In particular, in those typical cases where the periodic orbits are dense in the attractor, a typical trajectory can be viewed as a sequence of finite-time visits to neighborhoods of UPOs. 
As discussed in \cite{jenkinson2019ergodic}, UPOs support ergodic measures that often realize extreme behavior in a typical forced system.

We compute UPOs on the Lorenz attractor as described in the Appendix \ref{sec:PO_computation} for the subsequent analysis. A collection of them is shown in Figure~\ref{fig:UPOs} and some basic characteristics of these orbits are listed in Table~\ref{tab:POs}. In the next sections, we use these UPOs to explore the tipping window for a range of timescale ratios $\gamma$.

\subsubsection{The tipping window for limiting timescale separations}
\label{sec:DW_InfTimescaleSep}

As described in Theorem \ref{thm:limits}, in the limit of infinite timescale separation, the UPO forcing of the form considered here limits to a constant additive shift of the bifurcation parameter $\eta$. The location of the bifurcation then depends on the forcing strength $a$ and is given by that Theorem. We plot the bifurcation location for each of the $20$ UPOs for both limits of infinite timescale separation in Figure~\ref{fig:DWgamma_inf}. 

\begin{figure}
    \centering
    \includegraphics[align=t,width=7.5cm]{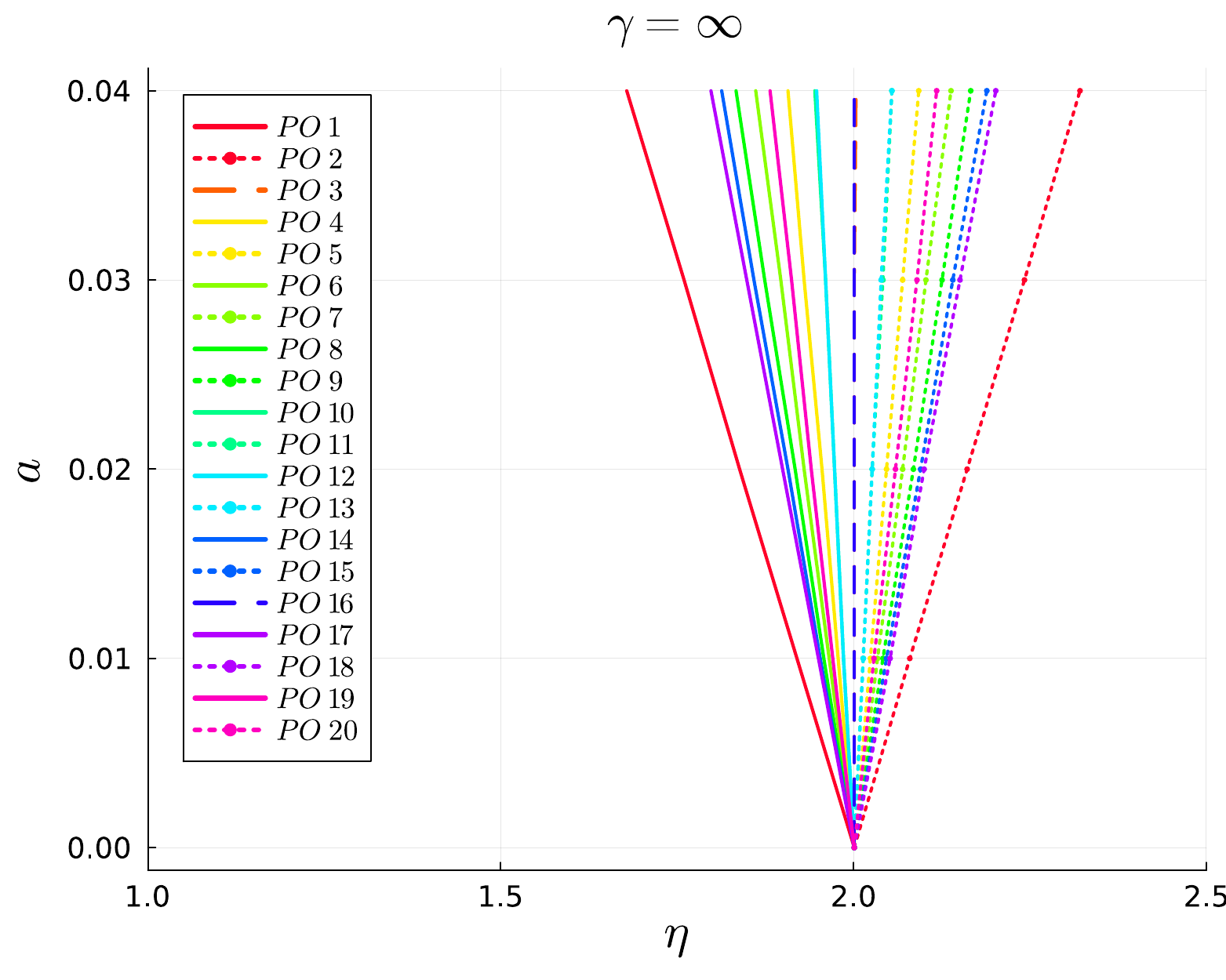} 
    ~\includegraphics[align=t,width=7.5cm]{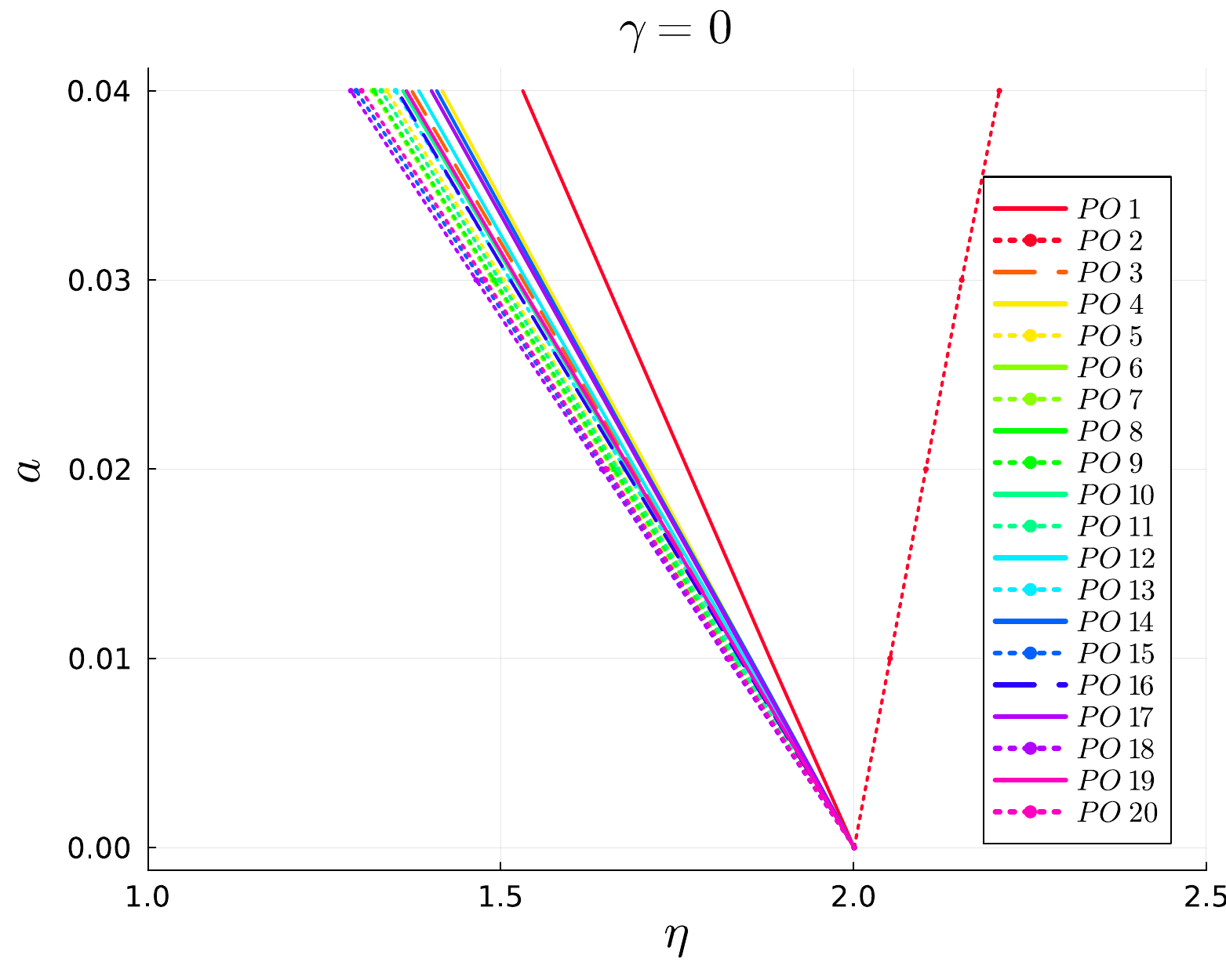}  
    \caption{The coloured lines show for each $a$ the lowest value of $\eta$ for which the $x$ dynamics (initialized in $x=-1.5$) fulfilled $x>0$ at some point during the integration time (i.e. the $x$ dynamics leave the vicinity of $x_-(\eta)$ and go to $x_+(\eta)$ - they ``tip''), when forcing the double-well dynamics given by Equation (\ref{eq:doubleWell}) with the mean (left plot) or the maximum (right plot) of the colour-coded UPO. The law of motion is of the form of Equation (\ref{eq:addforcedSyst}) with $\phi(y(\gamma t))=\bar{y}_{1,UPO_k}$ given by the mean (left plot) or $\phi(y(\gamma t))=\max(y_{1,UPO_k})$ the maximum (right plot) of the $k$-th UPO. The values $\bar{y}_{1,UPO_k}$ and $\max(y_{1,UPO_k})$ are given in Tab. \ref{tab:DWgrayShadingTimes}. For each UPO, we fix discrete values of $a\in[0,0.04]$ with spacing of $0.01$ and then do bisections in $\eta\in(0.4,1.6)$ to approximate the lowest $\eta$ up to an accuracy of $5\times10^{-3}$ for which the system, initialized in $x=-1.5$, tips, i.e. it fulfills $x>0$ during the integration time. Note that in the left plot, the lines are ordered from left to right by decreasing mean $y_1$-value of the associated UPOs, and each line is given by $\eta_k(a)=\eta^*-a\bar{y}_{1,UPO_k}$ as derived in section \ref{sec:inf_ts_sep}. In the right plot, the lines are ordered from left to right by decreasing  maximum $y_1$-value of the associated UPOs, and each line is given by $\eta_k(a)=\eta^*-a\max(y_{1,UPO_k})$.}
    \label{fig:DWgamma_inf}
\end{figure}

For $\gamma \rightarrow \infty$, the $y$-dynamics are infinitely faster than the $x$-dynamics, and Theorem \ref{thm:limits} implies that the lines are ordered by the mean value of the $y$ dynamics, $\bar{y}_1$. Forcing the double-well system by the self-symmetric UPOs 3 and 16 resulted in the dashed lines, which lie on top of each other because they have the same mean $y_1=0$.
The lines resulting from forcing with UPOs that are symmetrically related by a rotation around the $y_3$ axis, given by the mapping $(y_1,y_2,y_3)\mapsto(-y_1,-y_2,y_3)$, are mirrored along the axis $\eta = 2$. 

For the limit $\gamma \rightarrow 0$, the lines are ordered by the $y$-dynamics' maximum value $\max(y_1)$. Forcing by self-symmetric UPOs now does not result in the same dashed line since different symmetric UPOs can have different maximal $y_1$ values.  
In addition, the lines of the rotated periodic orbits are no longer mirrored along a line. Note that even though the lines of the 1st and 2nd UPOs are seemingly in the same positions as for $\gamma \rightarrow \infty$, they slightly shifted since they are now given by the maximum $\max(y_1)$ along the UPOs whereas they were given by the mean $\bar{y}_1$ for $\gamma \rightarrow \infty$.

\subsubsection{The tipping window for intermediate timescale ratios}
\label{sec:DW_ItermediateRelTimescales}

Now, we would like to understand the tipping window for cases other than limiting timescale separation. We simulate for four different values of the timescale parameter $\gamma$, forcing the double-well system with a) the UPOs from Figure~\ref{fig:UPOs} and b) with randomly chosen non-periodic chaotic trajectories on the Lorenz attractor. Note that care needs to be taken that the slowest system is integrated for a sufficiently long duration.

Figure~\ref{fig:DWgamma_intermediate} shows these results in a similar way to Figure \ref{fig:DWgamma_inf}: for the colored curves (resulting from UPO forcing - case a)), for each of the different values of the timescale parameter $\gamma\in(0.01, \, 0.1, \, 1, \, 10)$ in the four panels, we choose a range of different forcing amplitudes $a\in[0,0.04]$, and do a bisection in $\eta$ to find the best approximation of the lowest $\eta$ that leads to the $x$ dynamics tipping to the positive attractor $x_+(\eta)$, when initialized in $x=-1.5$, (i.e. in the basin of attraction of $x_-(\eta)$ if this attractor exists for the given value of $\eta$).

On the same figure, for a grid of $a$ and $\eta$ values, we run one simulation per grid point with the $x$-dynamics initialized in $x=-1.5$, and the $y$-dynamics initialized randomly close to the Lorenz-63 attractor with an initial transient removed. Then we monitor whether the trajectory tips to the other attractor $x_+(\eta)$ and, if so, at which time.

To interpret the four plots for different $\gamma$ in Figure \ref{fig:DWgamma_intermediate}, we think of fixing a forcing strength $a$. For small $\eta$, we do not observe tipping within the given time $t_{max}$ (black area). The larger $\eta$, the shorter the observed tipping times. The right end of the black area tilts to the left for smaller $\gamma$. We see that the cases $\gamma=0.01$ and $\gamma=10$ approximate well the limiting cases shown in Figure \ref{fig:DWgamma_inf}. For intermediate $\gamma$, a rather smooth interpolation between the two limiting cases can be seen.

\begin{figure}
    \centering
    \includegraphics[align=t,width=7.5cm]{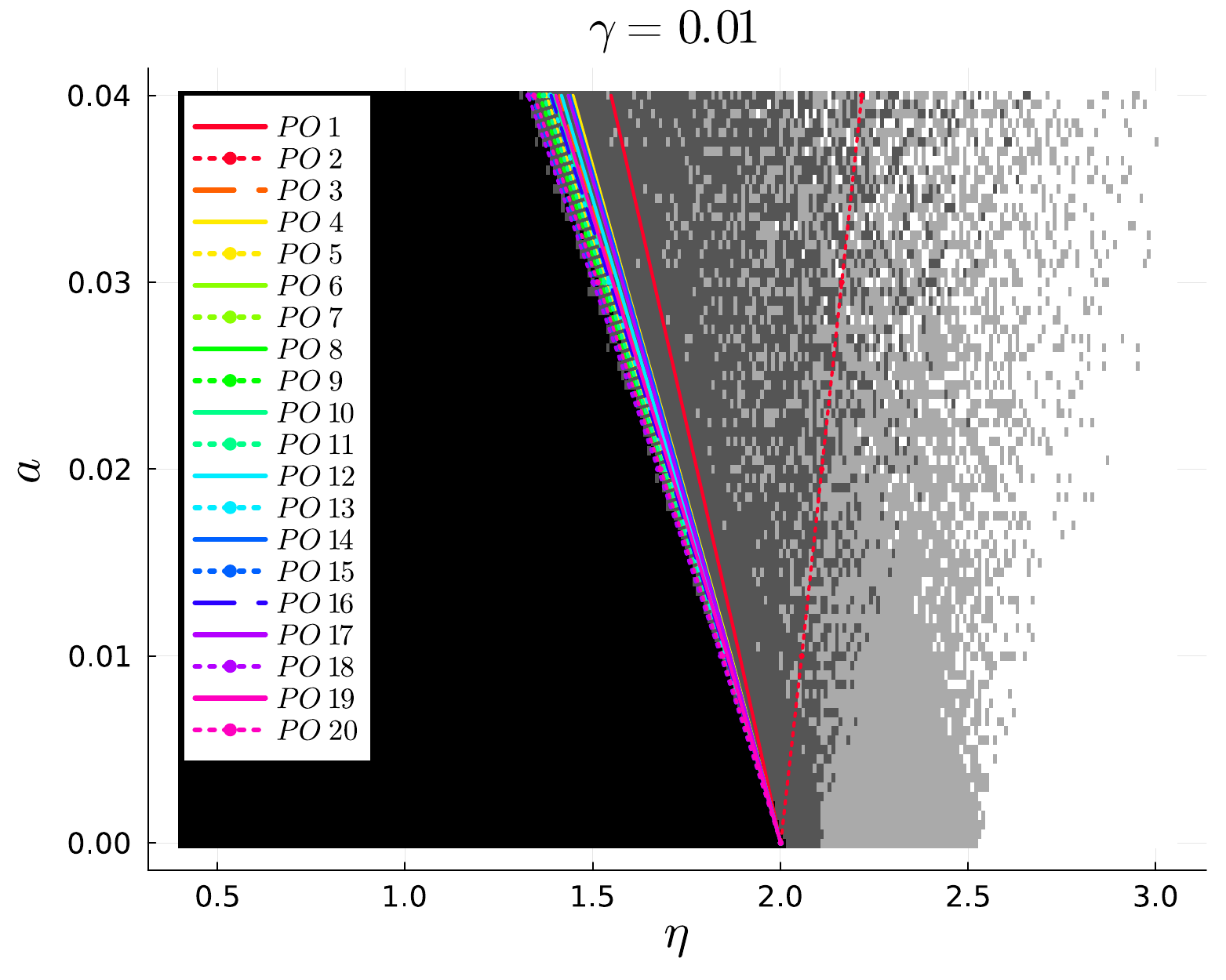}
    \includegraphics[align=t,width=7.5cm]{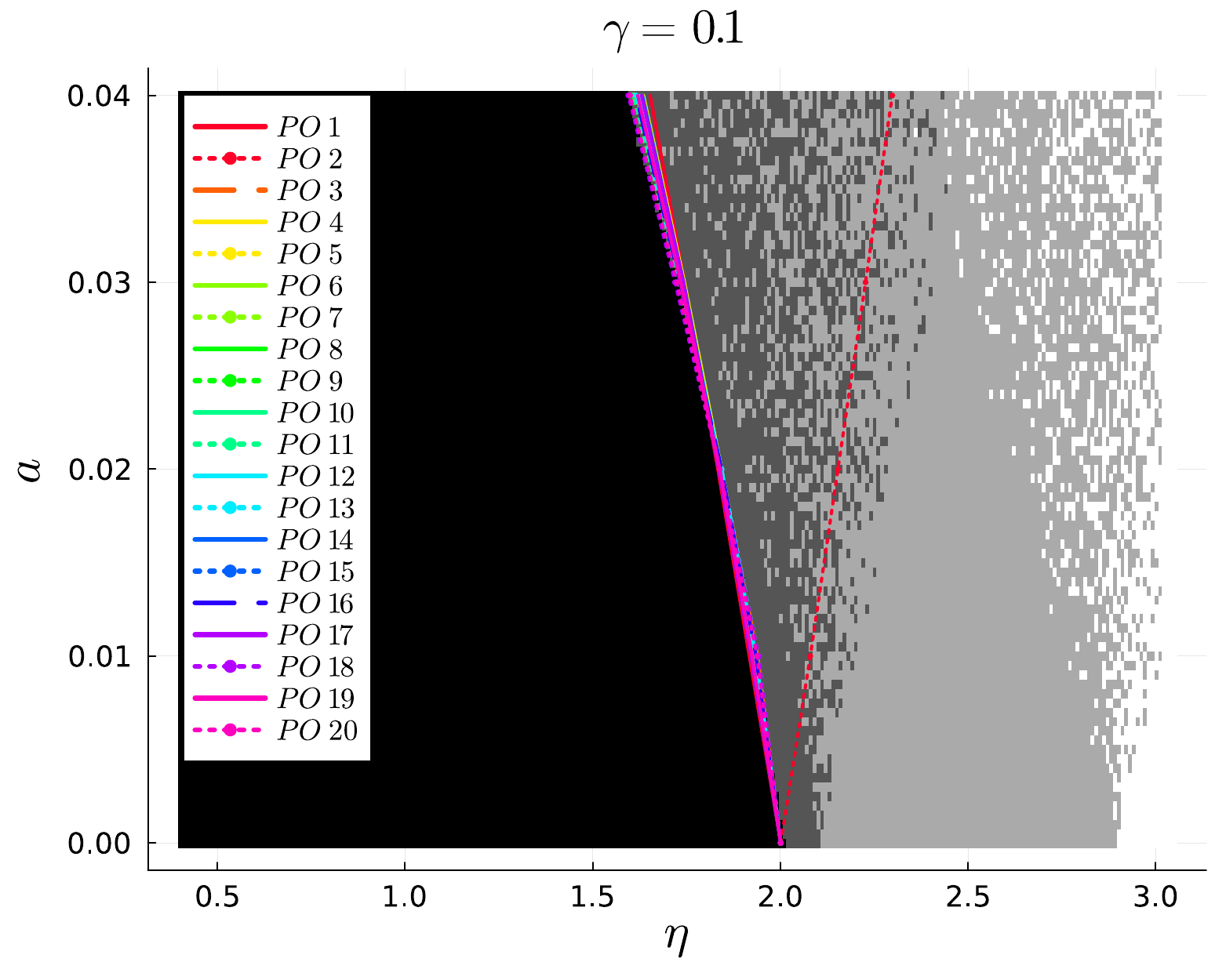}
    \includegraphics[align=t,width=7.5cm]{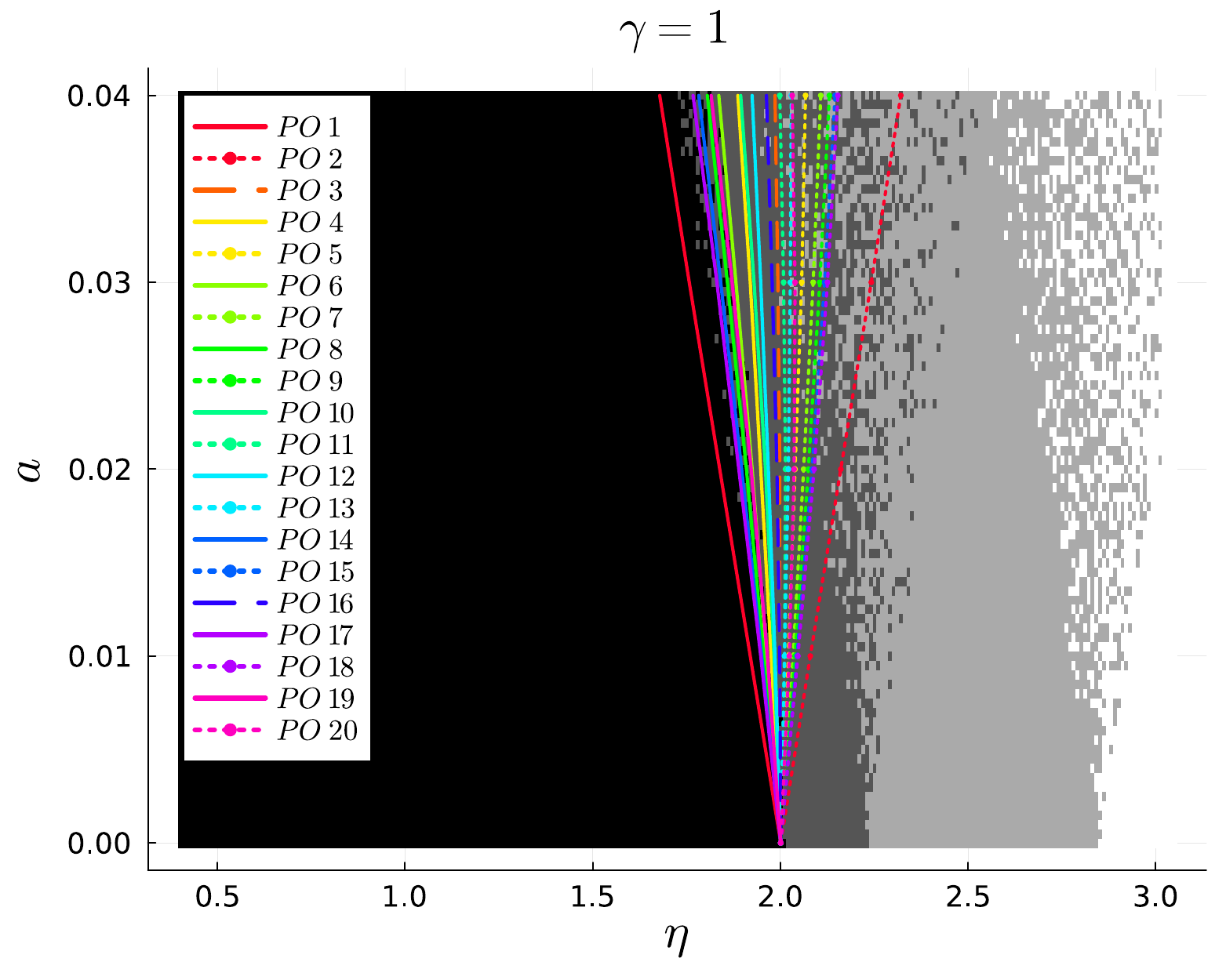}
    \includegraphics[align=t,width=7.5cm]{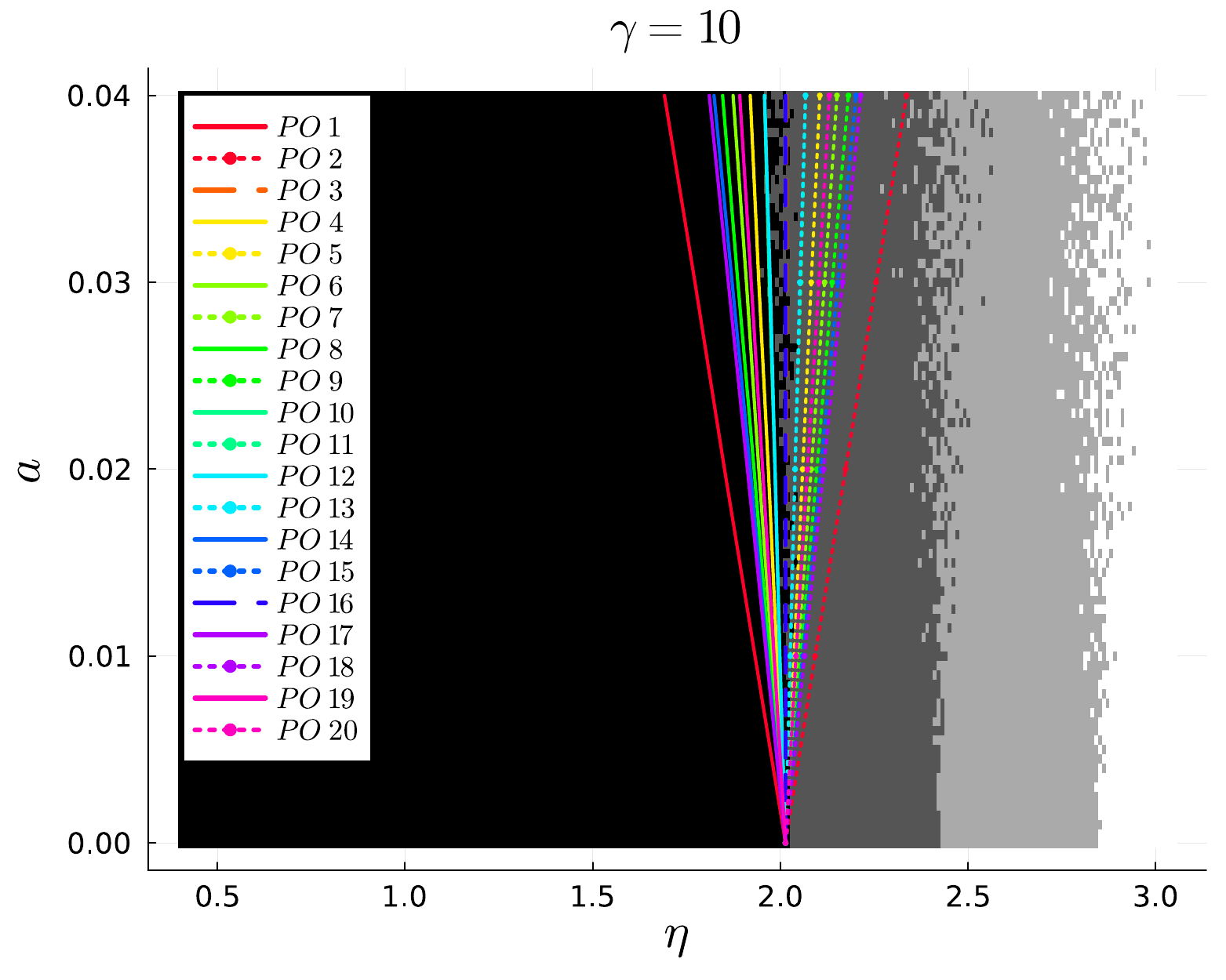}
    \caption{
    The colored lines show for each $a$ the lowest value of $\eta$ for which tipping is observed when forcing the double-well dynamics given by Equation (\ref{eq:doubleWell}) with the color-coded UPO. The system is of the form of Equation (\ref{eq:addforcedSyst}) with $y_1(t)=y_{1,UPO_k}$ given by the $k$th UPO and $\gamma\in(0.01, \, 0.1, \, 1, \, 10)$. For each UPO, we fix discrete values of $a\in[0,0.04]$ with spacing $0.01$ and then do bisections in $\eta\in(0.4,1.6)$ to approximate the lowest $\eta$ up to an accuracy of $5\times10^{-3}$ for which the system initialized in $x=-1.5$ tips to the positive attractor i.e. $x>0$ within the integration time.
    The black, dark gray, light gray, and white shading show in which time interval tipping of the double-well system \eqref{eq:LorDoubleWell} initialized in the basin of the attractor $x_-(\eta)$ (at $x=-1.5$), was observed as a result of forcing with a randomly chosen Lorenz-trajectory. The Lorenz initial condition of one run is given by letting the final condition of the previous run evolve for $5$ Lorenz-timesteps. The first initial condition of the Lorenz system is given by starting the system at $(y_1,y_2,y_3)=(0.1,0.1,25.1)$ and letting it relax to the attractor for $5$ Lorenz-timesteps.
    In the white areas, tipping to the other attractor was observed at a very small time, and in the black areas, tipping was not observed at all during the simulation time. The gray shadings show tipping for intermediate times. The time intervals corresponding to the different shading colors are given in Table \ref{tab:DWgrayShadingTimes} for each value of $\gamma$.}
    \label{fig:DWgamma_intermediate}
\end{figure}

\begin{table}
    \centering
    \begin{tabular}{c|c|c|c|c}
        $\gamma$& white    & lightgray & darkgray & black \\ \hline
        0.01   &[0,\, 0.0175)&[0.0175,\, 0.05)&[0.05,\, 150)&$[150,\, \infty)$ \\ \hline
        0.1    &[0,\, 0.12)&[0.12,\, 0.5)&[0.5,\, 150)&$[150,\, \infty)$ \\ \hline
        1 	   &[0,\, 1.25)&[1.25,\, 3.0)&[3.0,\, 150)&$[150,\, \infty)$ \\ \hline
        10     &[0,\, 12)&[12,\, 20)&[20,\, 150)&$[150,\, \infty)$ 
    \end{tabular} 
    \caption{Time intervals used for the different gray shadings in Figure~\ref{fig:DWgamma_intermediate}; if a randomly chosen initial condition for $y$ (with the $x$ dynamics initialized in $x=-1.5$) results in tipping within one of the time intervals then the corresponding shading is used, otherwise we used black to indicate that tipping has not occured by time $t_{max}=150$.}
    \label{tab:DWgrayShadingTimes}
\end{table}

Note that for $\gamma=0.01$ some chaotically forced trajectories tipped during the given simulation time for $\eta$ smaller than the leftmost UPO line. This is not observed for large $\gamma$. 
This can be expected from Equations \eqref{eq:limsTipWin_gamma0} and \eqref{eq:limsTipWin_gammaInf} for the {\bf lower bound of the tipping window:} 
\begin{itemize}
    \item For $\gamma$ small, the maximal $y_1$-values along the entire length of the orbits are relevant for the lower bound of the tipping window. A randomly chosen long chaotic trajectory on the Lorenz attractor is likely to assume larger values of $y_1$ than one of the $20$ considered UPOs. 
    \item For $\gamma$ large, the mean of the forcing trajectory is the relevant quantity. As the mean of a randomly chosen chaotic Lorenz trajectory is likely to be close to zero, more extreme forcing can be seen from UPO forcing, as UPOs are very unusual trajectories with a mean bound away from zero if they are asymmetric.
\end{itemize}

Approximating the {\bf upper bound of the tipping window} is easier using UPOs both for $\gamma$ small and for $\gamma$ large:
\begin{itemize}
    \item For $\gamma$ small, the minimal $y_1$-values along the entire length of the orbits are relevant for the upper bound of the tipping window. A randomly chosen long chaotic trajectory on the Lorenz-63 attractor is likely to explore both ``wings'' of the Lorenz-63 attractor (i.e. assumes both positive and negative $y_1$ values), while some UPOs (e.g. UPO 2) only explore negative values of  $y_1$. 
    \item For $\gamma$ large, the mean of the forcing trajectory is the relevant quantity. As the mean of a randomly chosen chaotic Lorenz trajectory is likely to be close to zero, more extreme forcing can be seen from UPO forcing.
\end{itemize}

Thus, the only case where chaotic forcing should be used to approximate the boundary of the tipping window in the Lorenz-double well system is when we are interested in the lower bound of the tipping window for small $\gamma$.

\begin{comment}
The second point also implies that the upper bound of the tipping window can be better approximated with the UPOs for large $\gamma$ as some UPO forcing trajectories can have a large negative mean value of $y_1$ and thus lead to tipping for unusually large $\eta$. In contrast, it is unlikely to find long chaotic forcing trajectories that stabilize the $x$ dynamics enough for them not to tip for large $\eta$.  
The upper bound of the tipping window for $\gamma$ small is also unlikely to be well approximated by a randomly chosen chaotic forcing trajectory, as such a chaotic trajectory is unlikely to have a maximal $y_1$ that is smaller than the maximal $y_1$s of every single UPO. In fact, UPO2 only explores a subset of the negative wing of the Lorenz attractor, and some of the other UPOs do not explore large positive values of the positive wing of the Lorenz attractor.  
Thus, the only case where chaotic forcing should be used to approximate the boundary of a tipping window is the lower bound of the tipping window for small $\gamma$.

As discussed in Section \ref{sec:inf_ts_sep}, in the limit of infinitely fast periodic forcing, the forced system only sees the mean of the forcing. In the opposite limit, the forcing is given by just the initial condition of the periodic orbit. 
\end{comment}

When fixing the forcing strength to some value $a=0.04$, we can plot $\gamma$ vs. $\eta$  and see how the intermediate cases interpolate between the limits of infinite timescale separation shown in Figure \ref{fig:DWfixed_a}.

\begin{figure}
    \centering
    \includegraphics[align=t,width=12.5cm]{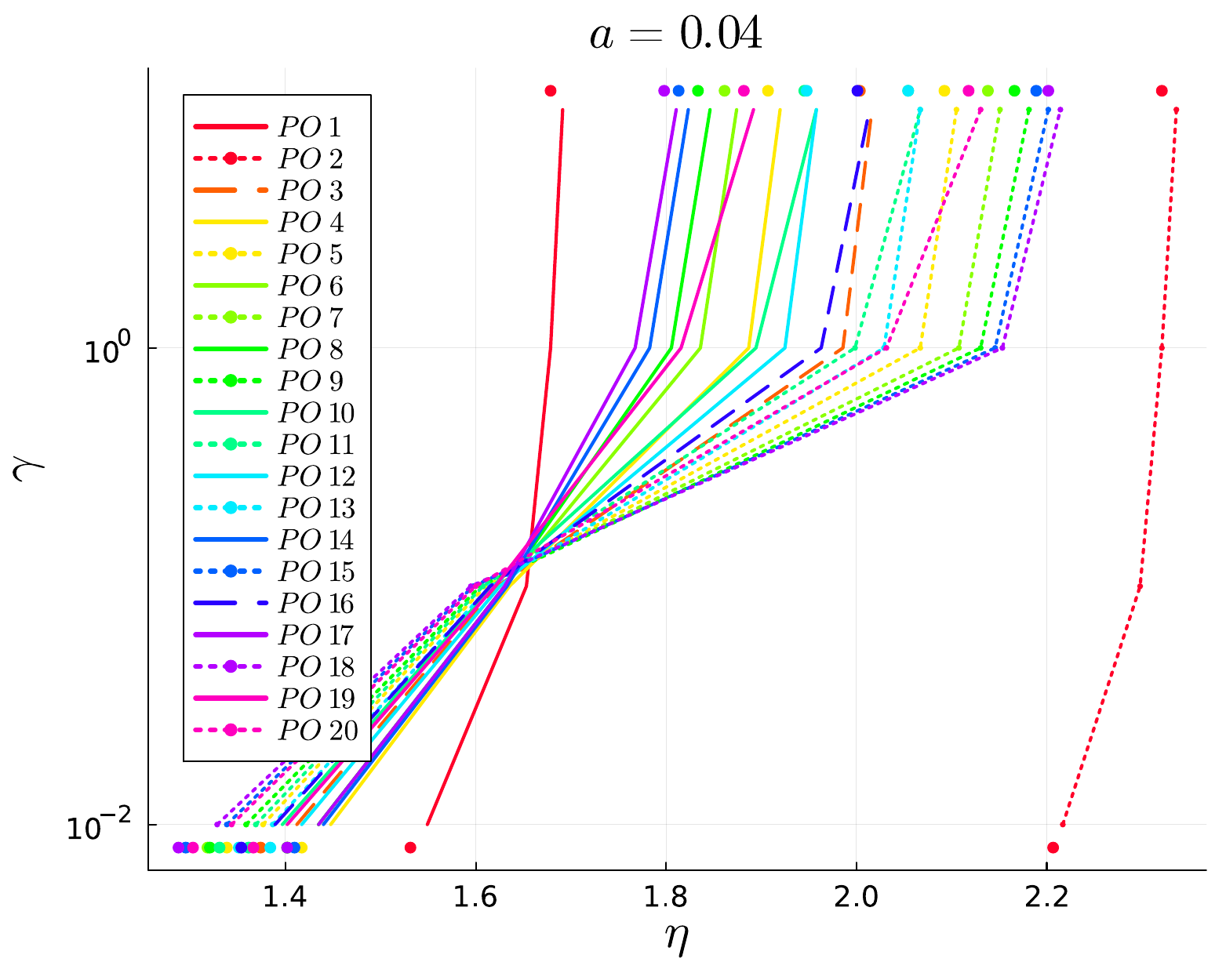}
    \caption{For fixed $a=0.04$, we plot the critical values of $\eta$ from Figure \ref{fig:DWgamma_intermediate} vs. $\gamma$. The dots at the top and bottom show the $\eta$ value at tipping of the tipping of the $x$ dynamics for the limiting cases of timescale separation.}
    \label{fig:DWfixed_a}
\end{figure}

\subsubsection{Crossing of UPO lines and the lower bound of the tipping window}
\label{sec:CrossingAndLowerBound}

In Figure \ref{fig:DWgamma_intermediate}, we can see that for $\gamma=0.1$, the lines resulting from forcing with UPOs $3$ to $20$ cross the UPO1 line. The location of the crossing points of each single UPO3-20 line with the UPO1 line depends on $a$, $\eta$, and $\gamma$.

We want to approximate the location of the crossing point of the UPO1 and UPO4 lines for different values of $\gamma$. 
First, we construct an approximation of the line resulting from UPO1 forcing that approximately holds for all values of $\gamma$. It lies exactly between the UPO1 lines of the limiting cases $\gamma=0$ and $\gamma=\infty$, and is given by $\eta(a)=\eta^*-\frac{a}{2}(\max(y_{1,UPO_1}) + \bar{y}_{1,UPO1} )$. Then, we fix a value of $\gamma$ and approximate the smallest value of \ $a$ \ for which the Lorenz-double-well system's $x$-dynamics would tip in response to UPO4 forcing, when $\eta(a)$ is given by this approximation of the UPO1 line. We find this value using a bisection in \ $a$. Figure \ref{fig:DWgamma_a_scaling} shows that the location where the UPO1 and UPO4 lines cross, scales approximately quadratically as $a \propto \gamma^2$.

\begin{figure}
    \centering
    \includegraphics[align=t,width=14cm]{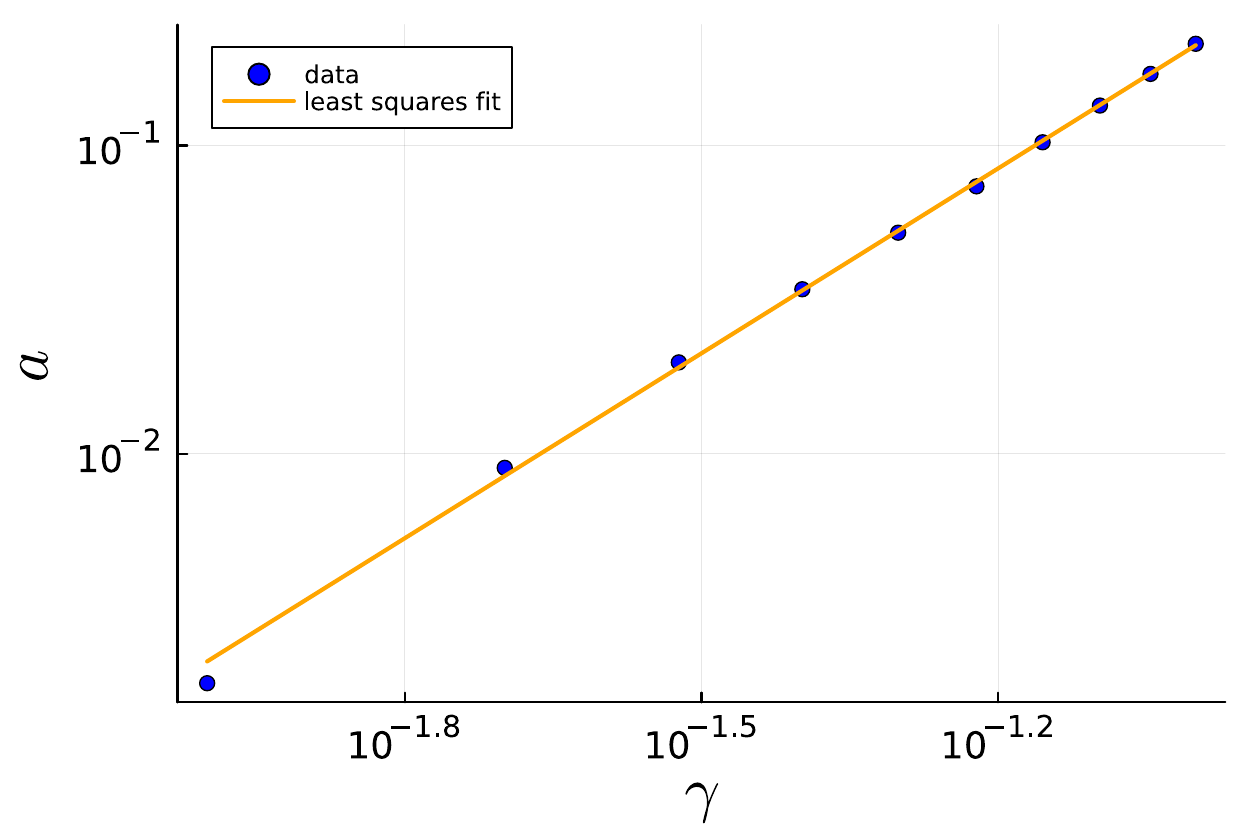}  
    \caption{For a set of fixed $\gamma\in\{0.01,0.02,\ldots,0.1\}$, we approximate the value of \ $a$ \ for which tipping of the Lorenz-double-well system's $x$ dynamics in response to UPO4 forcing occurs at the same value of $\eta$ as tipping in response to UPO1 forcing using a bisection method. We do this by fixing a $\gamma$ value, setting $\eta(a)=\eta^*-\frac{a}{2}(\max(y_{1,UPO_1}) + \bar{y}_{1,UPO1} )$ (which is an approximation of the $\eta$ value for which the system tips in response to UPO1 forcing), and then doing a bisection in \ $a$ \ to find the smallest value of \ $a$ \ for which the $x$ dynamics tip in response to UPO4 forcing. This gives \ $a$ \ as a function of $\gamma$ , shown here in blue in the log-log plot. The orange line corresponds to a least-squares fit of the function $a=c\gamma^2$ to the blue dots with best fit $c= 21.15$ and standard error $\pm 0.00129$.}
    \label{fig:DWgamma_a_scaling}
\end{figure}

This seems reasonable, as for intermediate timescale ratios, the forcing by a UPO may overshoot the threshold $\eta$ for a short time, and as long as this is not too long relative to the time spent below the threshold, the system may still not tip.

In this context, we note that an inverse square law determines whether a temporary overshoot over a fold bifurcation leads to tipping to the other attractor or not \cite{ritchie2019inverse}; this is given by $R t_{over}^2<\kappa$ with $R$ the maximal value of the overshoot, $t_{over}$ the total time during which the parameter is larger than the bifurcation value, and $\kappa\in\mathbb{R}$ a constant plus an additive correction of the order of the timescale of the parameter ramping.
By using the relation $\eta(a)=\eta^*-\frac{a}{2}(\max(y_{1,UPO_1}) + \bar{y}_{1,UPO1} )$, we can assume the overshoot time $t_{over}$ to be inversely proportional to $\gamma$ and the maximal overshoot distance $R$ to be proportional to the forcing amplitude $a$. This suggests that the critical value of $a$ scales as
$a\approx \kappa\gamma^2$ which agrees well with the numerics shown in Figure \ref{fig:DWgamma_a_scaling}.

\subsection{Example II: 2D-Stommel model forced by the Lorenz-63 system}
\label{sec:StommelPOs}

Our analysis so far has been concerned with the case of one-dimensional dynamics forced by a one-dimensional projection of a chaotic system. In most real-world applications, the system that is forced by another chaotic system is higher-dimensional, though. 
In this section, we explore a simple two-dimensional model inspired by climate science to check that the qualitative conclusions also apply in a higher-dimensional system.

We consider a system inspired by \cite{ashwin2021physical}, which consists of a Lorenz-63 system forcing a Stommel model. The Stommel model is a strongly reduced model of the Atlantic meridional overturning circulation (AMOC) \cite{stommel1961thermohaline}, whose two state variables $x\in\bbR^2$ describe the salinity and temperature difference between two ocean regions \cite{Dijkstra:2013}. 

Forcing the Stommel system with the Lorenz-63 system (with the typical parameters in the chaotic regime) can be seen as a generic, bounded chaotic forcing of the Stommel model and as a simple way of modelling internal variability; the timescales may or may not be separated depending on the choice of $\gamma$.
Note that we do not claim the considered parameters and model architecture to be of direct relevance to the real-world AMOC.
The model equations for $x=(x_1,x_2)$ and $y=(y_1,y_2,y_3)$ are given by
\begin{equation} \label{eq:Lorenz-Stommel}
\left.\begin{aligned}
    \frac{d}{dt}{x} &= f(x,\xi_1+a y_1, \eta_1 + a y_1)  \\
    \frac{d}{dt}{y} &= \gamma g(y)
    \end{aligned}\right\}
\end{equation}
where the Stommel right-hand side is
\begin{equation} \label{eq:Stommel}
\left.\begin{aligned}
    f_1(x,\xi_1+a y_1, \eta_1 + a y_1) &:= - x_1 (1 + | x_1 - x_2|)+\xi_1+a y_1 \\\
    f_2(x,\xi_1+a y_1, \eta_1 + a y_1) &:= - x_2(\zeta + |x_1 - x_2|) +\eta_1 + a y_1 \\
    \end{aligned}\right\}
\end{equation}
with $\zeta = 0.3, \, ~\xi_1 = 3$, and $\eta_1 \in [0,1.8]$ and the Lorenz right hand side is
\begin{equation}\label{eq:Lorenzz}
\left.\begin{aligned}
    g_1(y) &:= \sigma (y_2 - y_1) \\
    g_2(y) &:=y_1(\rho-y_3)-y_2 \\
    g_3(y) & := y_1 y_2 - \beta y_3.
\end{aligned}\right\}
\end{equation}
with the standard parameters $\sigma = 10, ~ \rho = 28,$ and $\beta = 8/3$. We vary $\gamma$, $a$ and $\eta_1$ in the numerical experiments and the timescale $t$ is that of the response system, Equation (\ref{eq:addforcedSyst}). 

In this model, the AMOC-strength can be defined as the difference of the Stommel variables \cite{Dijkstra:2013} as\footnote{This expression means that the AMOC strength is proportional to the positive temperature gradient and to the negative salinity gradient.}
\begin{equation}
    \Psi = x_1-x_2.
\end{equation}
The unforced Stommel model (\ref{eq:Stommel}) (i.e. the $x$-dynamics from (\ref{eq:Lorenz-Stommel}) with $a=0$) has a stable attractor, the AMOC-on state with large $\Psi_+(\eta)$ for $\eta<\eta^*$, a saddle-node bifurcation on increasing $\eta$ at $\eta^*\approx1.22$ and for $\eta>\eta^*$, the system has only one attractor; the AMOC-off state with small $\Psi_-(\eta)$. We remark that \cite{AxelsenQuinn2024} study attractor crises in this model, though they do not consider the effect of varying the timescale ratio between forcing and response system.

For $a>0$, the Stommel system is forced by the first chaotic variable $y_1$, so we are in a higher dimensional version of the setting (\ref{eq:addforcedSyst}). Note that the forcing does not only act on the bifurcation parameter $\eta_1$, which appears in the second Equation of the Stommel system \eqref{eq:Stommel}, but it also acts on the dynamics of the first Stommel variable. In Figure \ref{fig:LStyptraj}, we see typical trajectories showing the AMOC strength.

\begin{figure}
    \centering
    \includegraphics[align=t,width=14cm]{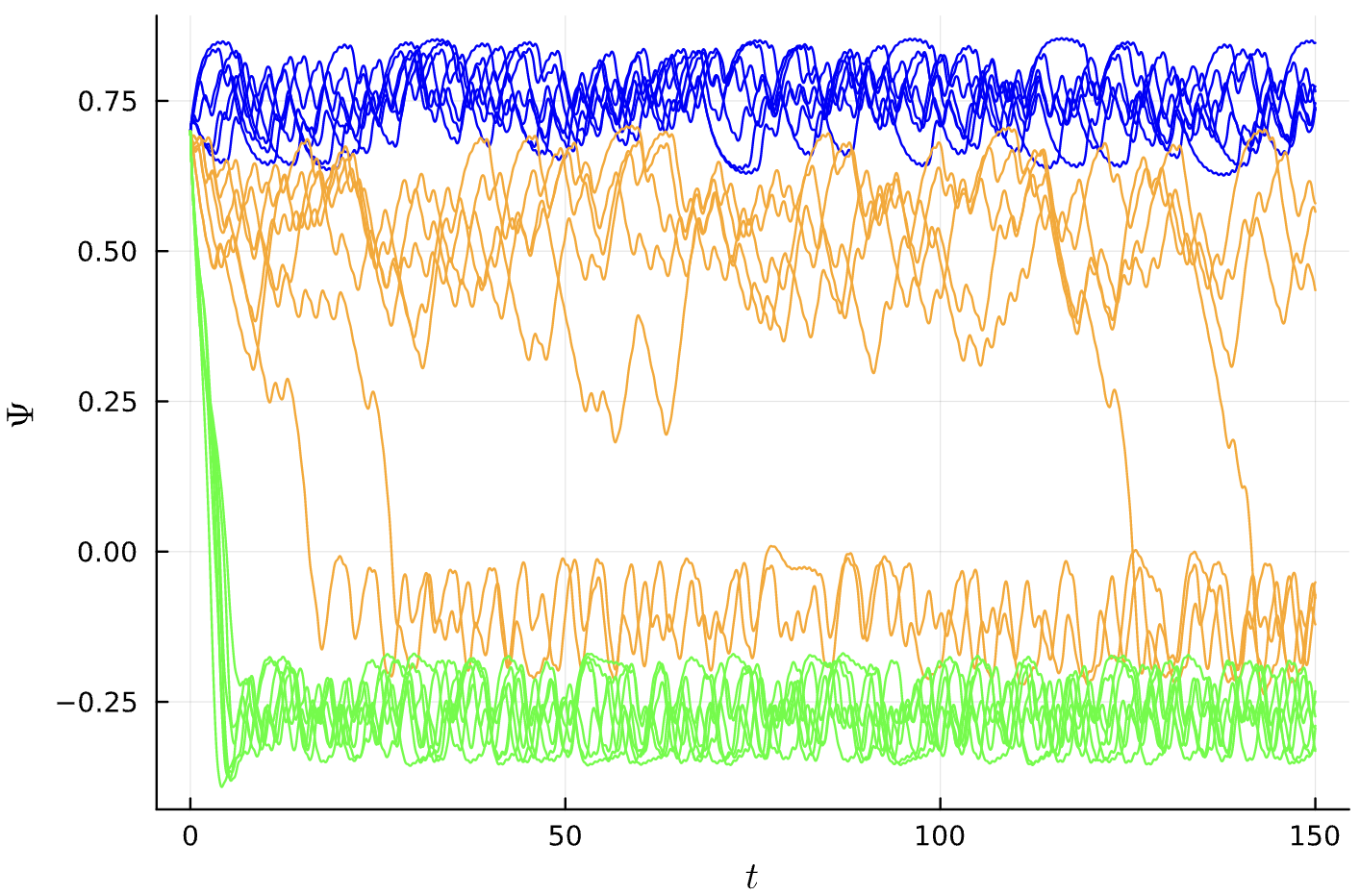}  
    \caption{Typical trajectories of the Lorenz-Stommel system (\ref{eq:Lorenz-Stommel}) with $a=0.04$, $\gamma=1$, randomly chosen initial conditions of the Lorenz-$y$-dynamics, and three different values of $\eta_1$. The $x$-dynamics are initialized in the vicinity of the AMOC-on attractor at $(x_1,\,x_2) = (1.7,\, 1.0)$. The blue curves ($\eta_1=1.0$) remain in the vicinity of the on-attractor of the response system. The orange curves ($\eta_1=1.175$) show tipping within the given time depending on the specific $y$-trajectory. Eventually, they would all go to the off attractor as they would leave the zero-measure set of remaining orbits (if these still exist for $\eta_1=1.175$) due to numerical errors. The green curves ($\eta_1=1.5$) all tip independently of the specific Lorenz trajectory. The $y$-trajectories are generated by taking the final state of the previous ensemble member as the initial condition for the next one.}
    \label{fig:LStyptraj}
\end{figure}

To approximate the chaotic tipping window, we need to understand for which values of $a$ and $\eta_1$ the Stommel system initialized in the basin of the AMOC-on state tips to the AMOC-off state. 
Therefore, we perform simulations similar to those in Section \ref{sec:DoubleWellLorenz} both for limiting timescale separation and intermediate timescale ratios.

\subsubsection{The tipping window for limiting timescale separations}\label{sec:LS_InfTimescaleSep}

First, we consider the two limiting cases of infinite timescale separation $\gamma=0$ and $\gamma=\infty$. 
We initialize the system in the vicinity of the AMOC-on state and force it with either the means of the first UPO coordinates (for $\gamma=\infty$) or with the maxima of the first UPO coordinates (for $\gamma=0$) and do a bisection in $\eta_1$ to find the lowest $\eta_1$ for which the system tips. The results are shown in Figure \ref{fig:LSgamma_lims}. We observe a very similar behavior to the case of the double-well system shown in Figure \ref{fig:DWgamma_inf}.

\begin{figure}
    \centering
    \includegraphics[align=t,width=7.5cm]{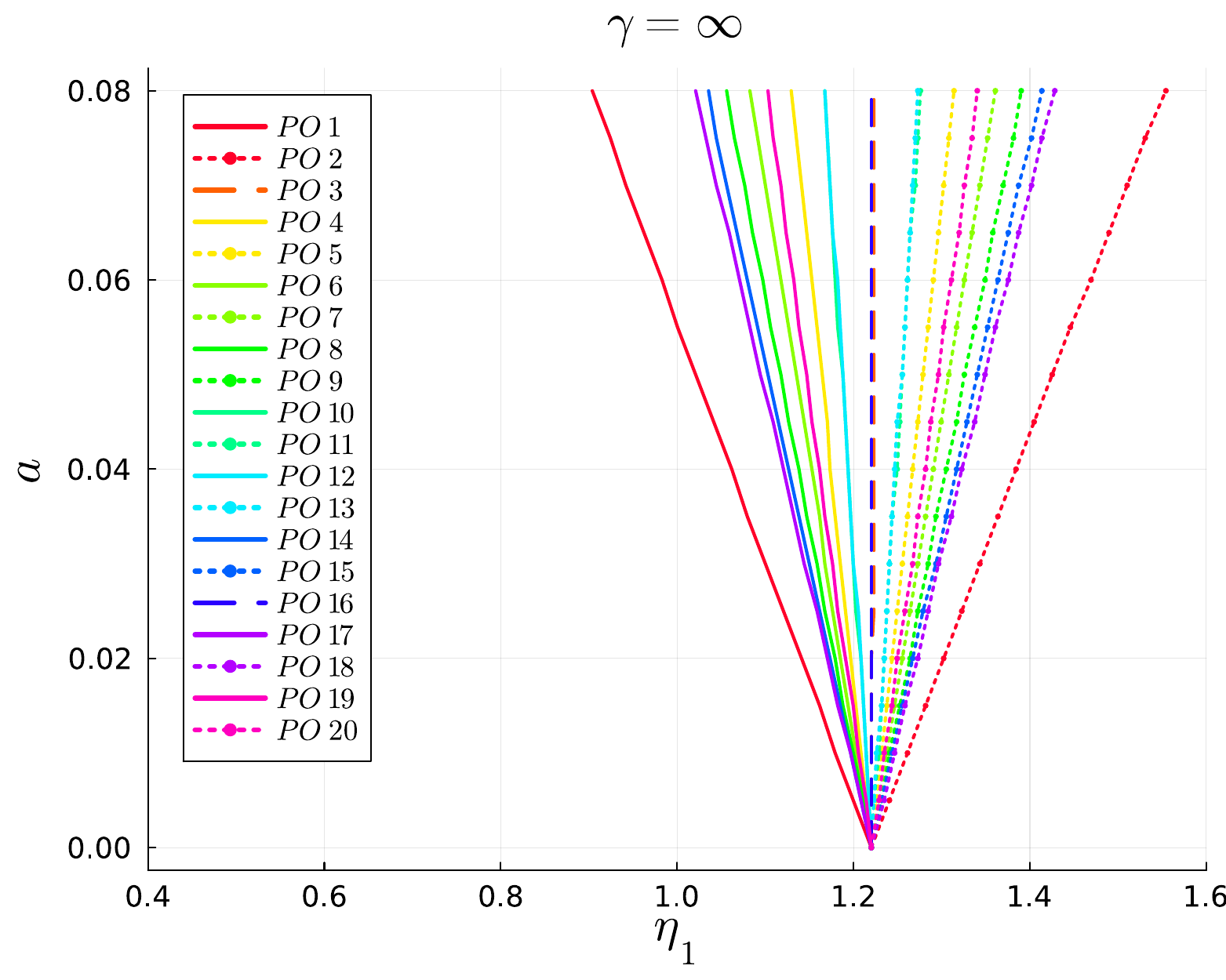} 
    ~\includegraphics[align=t,width=7.5cm]{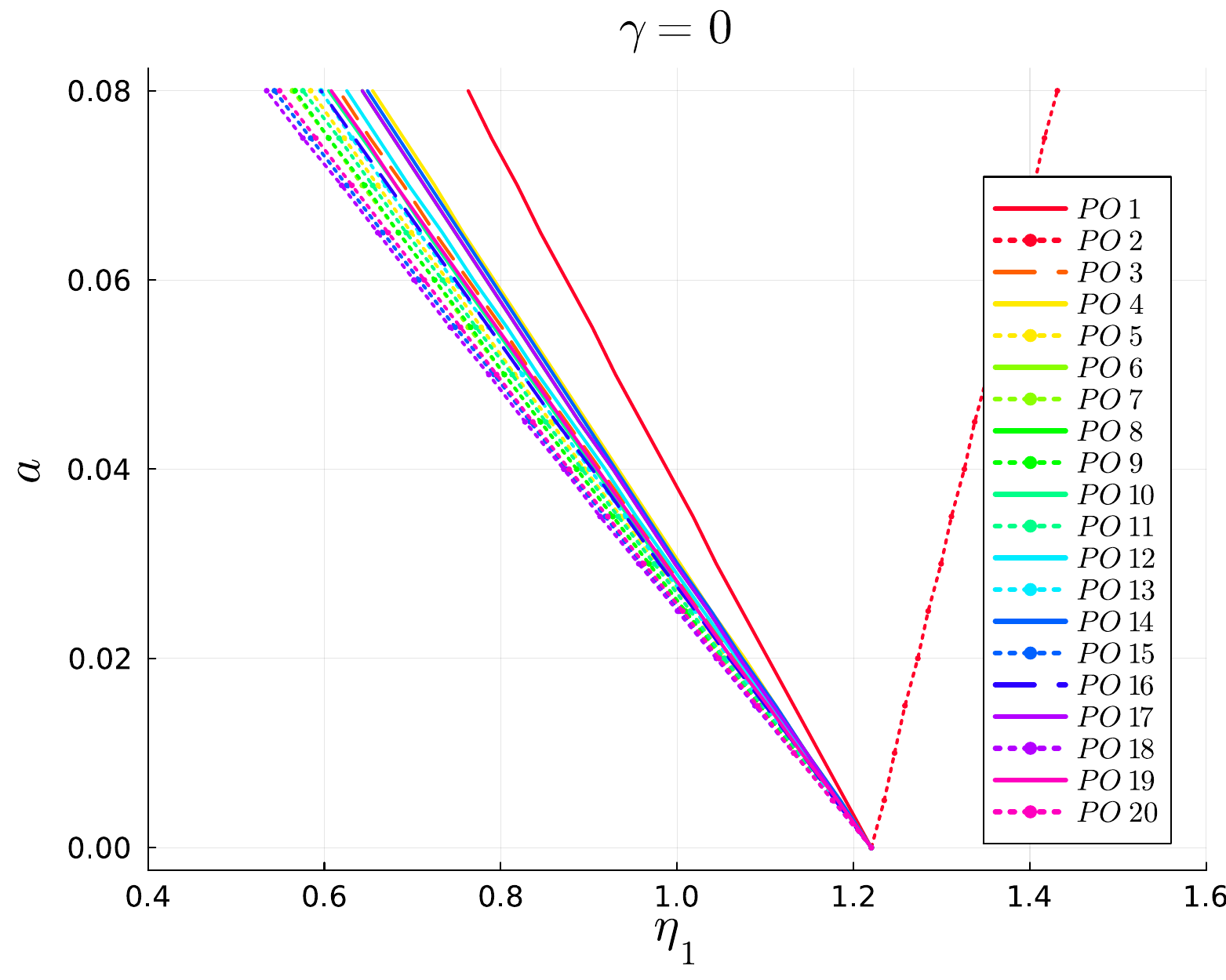}  
    \caption{The colored lines show for each $a$ the lowest value of $\eta_1$ for which tipping is observed when forcing the Stommel system (\ref{eq:Lorenz-Stommel}) with the mean $\bar{y}_{1,UPO_k}$ (left panel) or with the maximum $\max(y_{1,UPO_k})$ (right panel) of the first component of a colour-coded UPO from the Lorenz system. For each UPO, we fix discrete values of $a\in[0,0.08]$ with spacing $0.01$ and then perform bisections in $\eta_1\in(0.4,1.6)$ to approximate the lowest $\eta_1$ up to an accuracy of $5\times10^{-3}$ for which the system initialized in the vicinity of the AMOC-on state at $(x_1,x_2)=(1.7,1.0)$, tips to the AMOC-off state, i.e. $\Psi<0.1$. Note the similarity to Figure~\ref{fig:DWgamma_inf}. }
    \label{fig:LSgamma_lims}
\end{figure}

\subsubsection{The tipping window for intermediate timescale ratios}
\label{sec:LS_ItermediateRelTimescales}

Now, we turn to cases of intermediate timescale ratio. Again, these simulations are performed similarly to the cases a) and b) described in Section \ref{sec:DW_ItermediateRelTimescales}. 

In case b) (forcing with randomly chosen trajectories of the Lorenz-63 system), we run the system on a grid of $a\in(0,0.08)$ and $\eta_1\in(1,1.5)$,  always choose the initial condition of the Stommel system \eqref{eq:Stommel} to be in the vicinity of the AMOC-on state by setting $(x_1,x_2)=(1.7,1.0)$, and consider a randomly chosen trajectory on the Lorenz attractor. Then, we monitor the time when the system tips, if it tips at all during the simulation time, and show the results through the gray shading in
Figure \ref{fig:LSgamma_intermediate}. To ensure that the initial conditions of the consecutive runs are not correlated, we choose the final condition of the Lorenz system of the previous run, evolve it for 5 Lorenz time units, and then choose this state as the initial condition for the next run. 

Case a) (forcing with the UPOs on the Lorenz-63 attractor) is again performed similarly to the analysis of the double-well system. We initialize the system in the vicinity of the AMOC-on state, force it with the UPOs, and perform a bisection in $\eta_1$ to find the lowest $\eta_1$ for which the system tips.
Note that these plots look qualitatively similar to those for the double-well system in Figure \ref{fig:DWgamma_intermediate}.

\begin{figure}
    \centering
    \includegraphics[align=t,width=7.5cm]{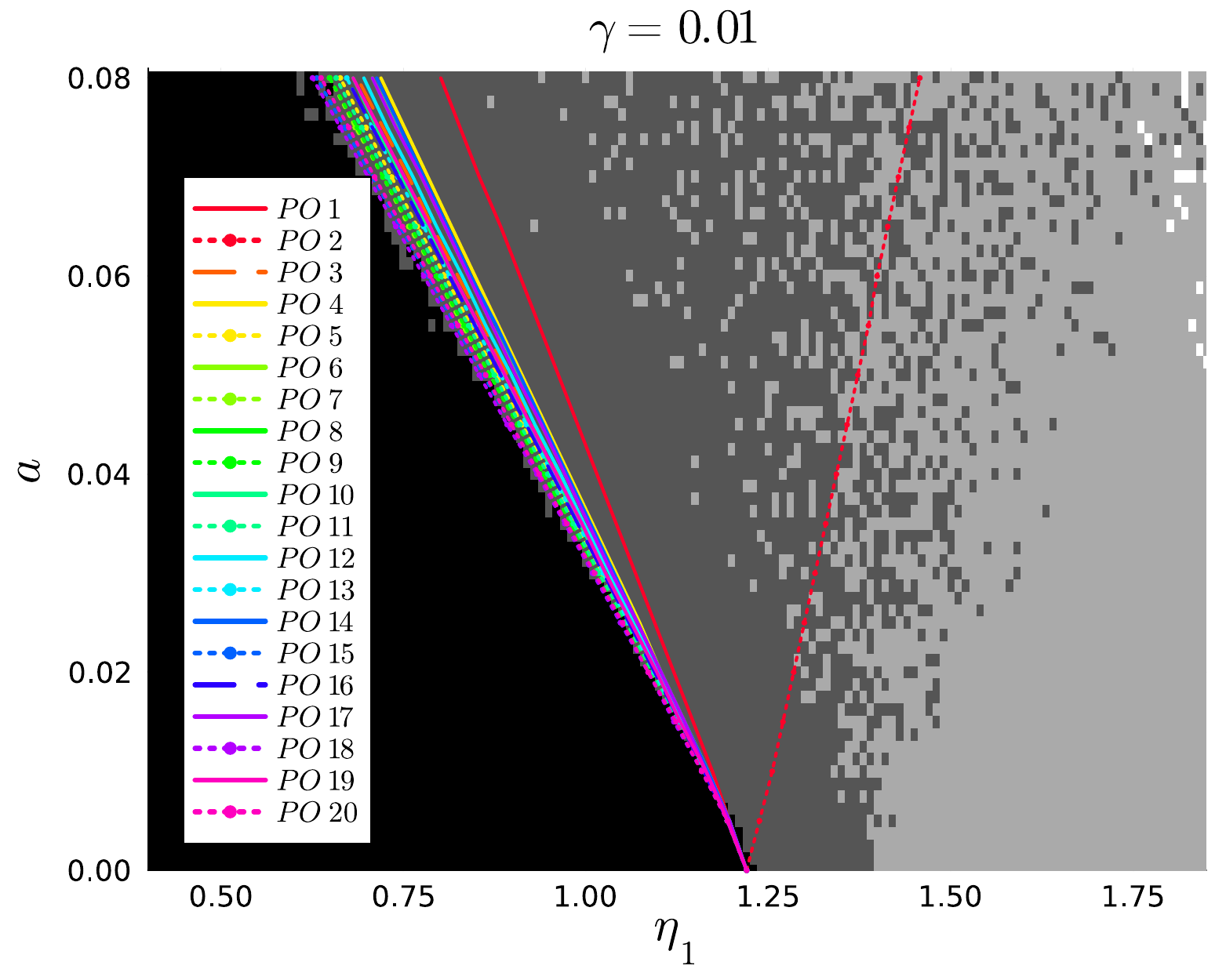}
    \includegraphics[align=t,width=7.5cm]{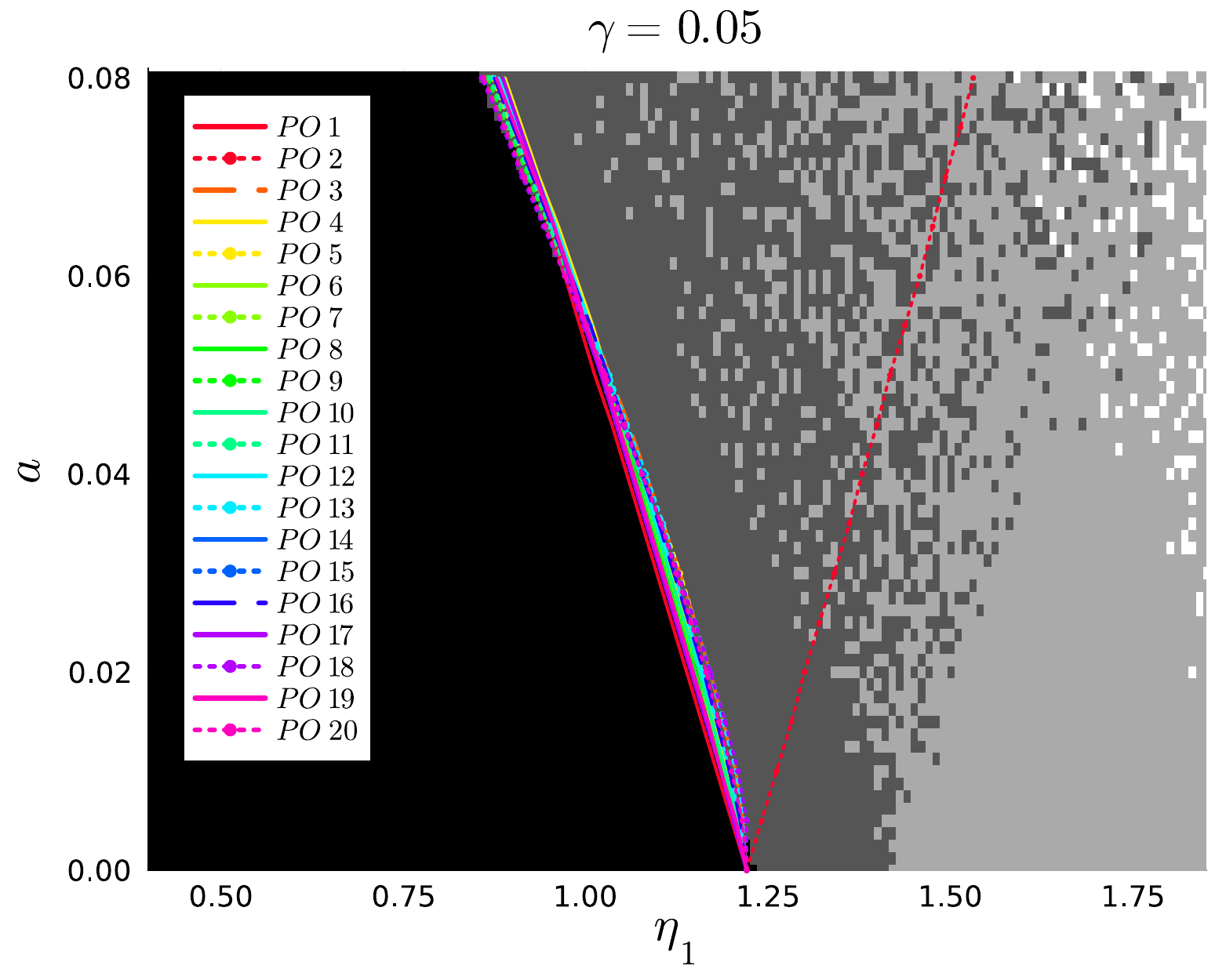}
    \includegraphics[align=t,width=7.5cm]{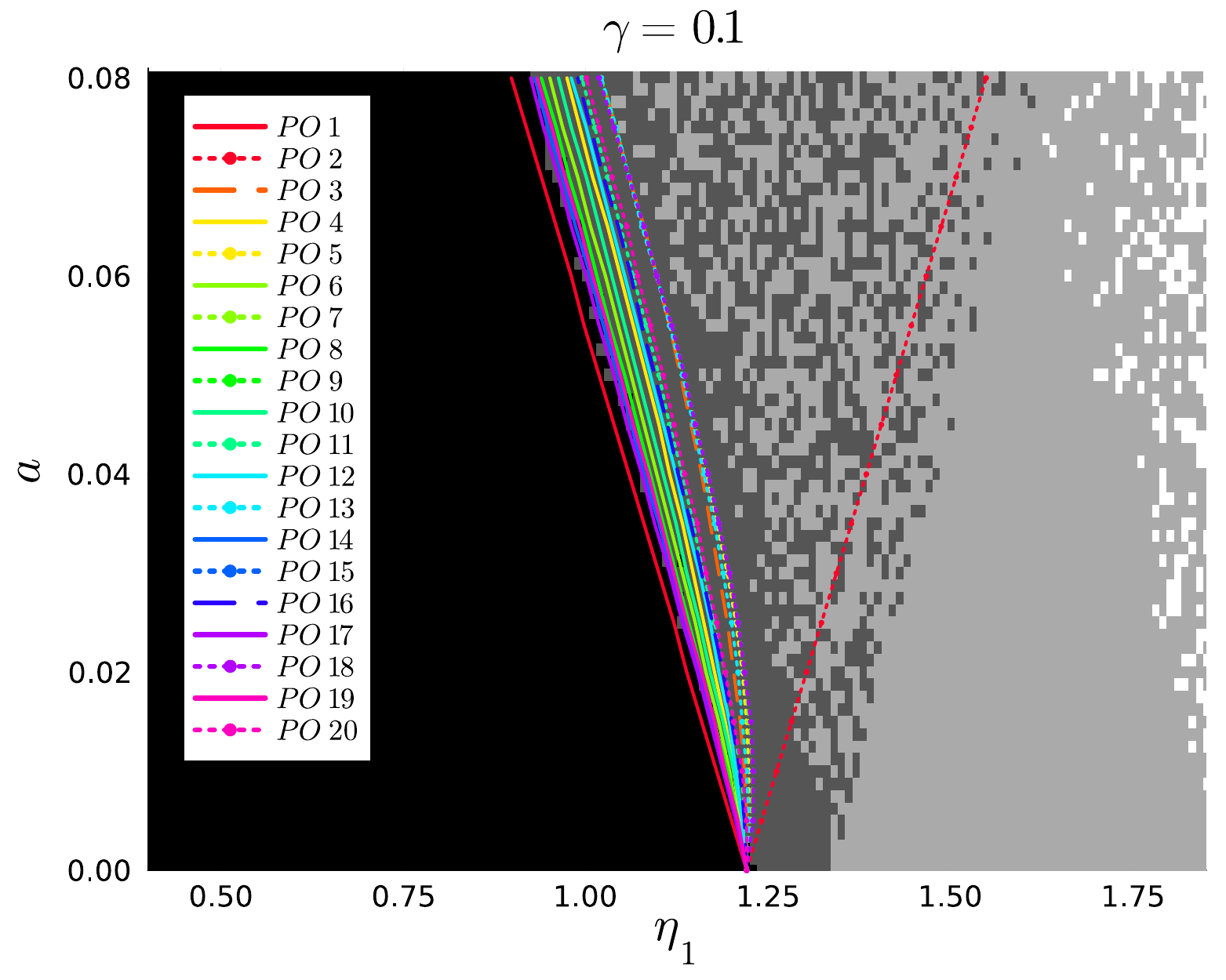}
    \includegraphics[align=t,width=7.5cm]{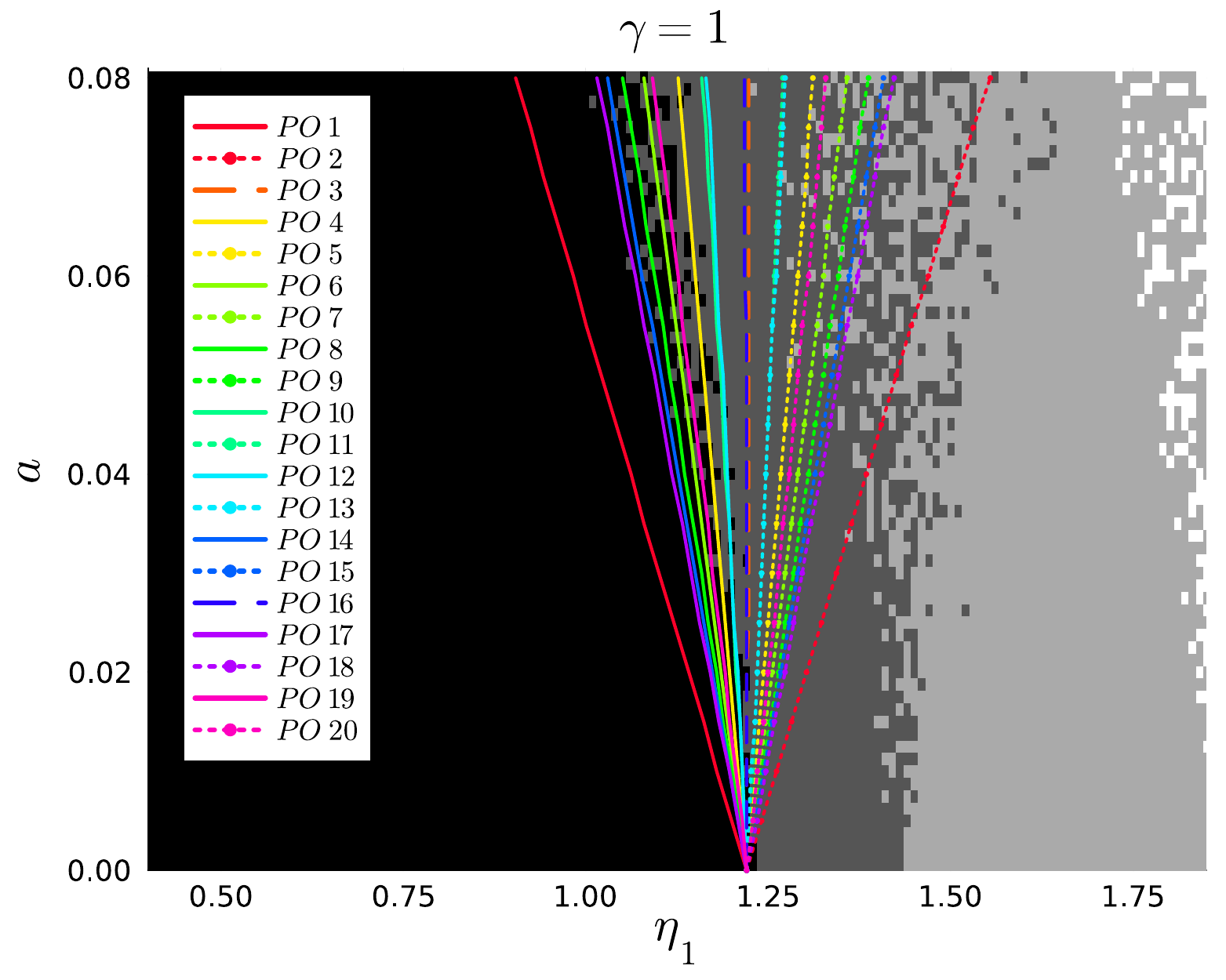}
    \caption{The colored lines show for each $a$ the lowest value of $\eta_1$ for which tipping is observed during the simulation time when forcing the Stommel system \ref{eq:Lorenz-Stommel} with the color-coded UPO. The system is a higher dimensional version of Equation \eqref{eq:addforcedSyst} with $y(t)=y_{UPO_k}$ given by the $k$th UPO and $\gamma\in(0.01, \, 0.05, \, 0.1, \, 1)$. For each UPO, we fix discrete values of $a\in[0,0.08]$ with spacing of $0.01$ and then perform bisections in $\eta_1\in(0.4,1.6)$ to approximate the lowest $\eta_1$ up to an accuracy of $5\times10^{-3}$ for which the system initialized in the vicinity of the AMOC-on state at $(x_1,x_2)=(1.7,1.0)$, tips to the AMOC-off state, i.e. $\Psi<0.1$.
    The black, dark gray, and light gray shadings are computed analogously to Figure \ref{fig:DWgamma_intermediate}: for each combination of $(\eta_1,a)$, we randomly chose a single orbit of the Lorenz system to force the Stommel system (\ref{eq:Lorenz-Stommel}) initialized in the vicinity of the AMOC-on state. The Lorenz initial condition of one run is given by letting the final condition of the previous run evolve for $5$ Lorenz-timesteps. The first Lorenz initial condition is given by starting the system at $(y_1,y_2,y_3)=(0.1,0.1,25.1)$ and letting it relax to the attractor for $5$ Lorenz-timesteps. 
    The gray shadings show tipping for intermediate times. The time intervals corresponding to the different shading colors are given in Table \ref{tab:LSgrayShadingTimes} for each value of $\gamma$.
    }
    \label{fig:LSgamma_intermediate}
\end{figure}

\begin{table}
    \centering
    \begin{tabular}{c|c|c|c|c}
        $\gamma$& white    & lightgray & darkgray & black \\ \hline
        0.01   &[0,\, 0.011)&[0.011,\, 0.05)&[0.05,\, 150)&$[150,\, \infty)$ \\ \hline
        0.1    &[0,\, 0.061)&[0.061,\, 0.22)&[0.22,\, 150)&$[150,\, \infty)$ \\ \hline
        1 	   &[0,\, 0.125)&[0.125,\, 0.7)&[0.7,\, 150)&$[150,\, \infty)$ \\ \hline
        5     &[0,\, 1.25)&[1.25,\, 4.0)&[4.0,\, 150)&$[150,\, \infty)$ 
    \end{tabular} 
    \caption{Time intervals during which tipping was observed for each of the different gray shadings in Figure \ref{fig:LSgamma_intermediate}, as in Table~\ref{tab:DWgrayShadingTimes}.}
    \label{tab:LSgrayShadingTimes}
\end{table}

Note that the simulations for Figures \ref{fig:LSgamma_lims} and \ref{fig:LSgamma_intermediate} were performed using the same initial condition of the Stommel system for all values of $\eta_1$, and we assume that this point remains in the basin of attraction of the AMOC-on-state whenever the on-state exists.
We have checked, and this assumption seems reasonable, even at the saddle node bifurcation.

The gray shadings in Figure \ref{fig:LSgamma_intermediate} suggest that some scaling laws can be found that (for fixed $a$) connect the expected tipping time to the distance between $\eta_1$ and the start of the tipping window $\eta_{1-}$. This is discussed in \cite{MEHLING2024134043, Boerner2025global, grebogi1985super, Grebogi:1986} and it would be interesting to see how the proposed scaling laws can be further related to the tipping window and how they depend on the forcing strength $a$.

Upon increasing $a$, the effect of the Lorenz forcing plays an increasingly important role, and the size of the chaotic tipping window increases with increasing $a$. The smaller $a$, the better the linear response theory, discussed in Section \ref{sec:inf_ts_sep}, describes the system's dynamics. Thus, it seems intuitive that for larger $a$, the ordering of the line by their mean $y_1$-value for large $\gamma$ starts to break first, and only for smaller $\gamma$ it also breaks down for small values of $a$.

\section{Chaotic tipping with drift: the dynamic tipping window}
\label{sec:DynamicTippingWindow}

We turn back to the general case of additively chaotically forced one-dimensional bistable dynamics, as in system (\ref{eq:frozenSystf0}). But now we consider the case where the bifurcation parameter $\eta_r(t)$ is time-dependent:
\begin{equation}
	\begin{aligned}  \label{eq:fullSystAdd}
		\frac{dx}{dt} = & f(x) + \eta_r(t) + a \phi(y) \\
		\frac{dy}{dt} = &\gamma g(y).
	\end{aligned}
\end{equation}
We define
\begin{align} \label{eq:eta_of_t}
    \eta_r(t) &:= \eta_0 + r t
\end{align}
with $r>0$ and $\eta_0$ real constants. The assumptions on the unforced system
\begin{align} \label{eq:addunforcedSystRate}
    \frac{dx}{dt} = &\ f(x)+\eta_r(0)
\end{align}
are the same as in the time-independent case described after Equation (\ref{eq:addunforcedSyst}).
We assume that the dynamics given by Equation (\ref{eq:addunforcedSystRate}) has a bifurcation point $\eta^*$ that can cause B-tipping. Specifically, suppose that there are two attractors for $\eta^{\dag}<\eta_0<\eta^*$, only one attractor for $\eta_0>\eta^*$, and that one of the attractors disappears at a saddle-node bifurcation $\eta_0=\eta^*$. 
The $y$-dynamics are assumed to be chaotic.
In the following, we assume $\gamma=1$, the forcing strength $a>0$ small enough, $\eta_0<\eta^*$, and we initialize the $x$-dynamics close to the $x$ dynamics' attractor that ceases to exist at $\eta^*$. 

In the limit $r\rightarrow 0$, the behavior of \eqref{eq:fullSystAdd} limits to that of System \eqref{eq:Lorenz-Stommel} described in the previous section. For $r>0$, $\eta_r(t)$ increases over time, and we expect a combination of non-autonomous effects and the dynamics described by the tipping window. 
Figure \ref{fig:LSdynTipWin} shows this for the Lorenz-Stommel system with an additional parameter ramping of the form of Equation \eqref{eq:eta_of_t}. The larger the value of $r$, the larger the value of $\eta_r(t_c)$ that is recorded when the AMOC strength $\Psi$ crosses the tipping threshold $\Psi=0.1$ at time $t_c$. This shift to larger values of $\eta_r(t_c)$ is observed as the forced system responds with a delay compared to the increase in $\eta_r(t)$, and the relaxation to the AMOC-off state corresponding to the current value of $\eta_r(t)$ needs some time during which $\eta_r(t)$ still increases. 

This effect is called a {\it dynamic tipping window} in \cite{AshwinNewmanRoemer:2024} where it is noted that the $\eta$-interval depends on $r$ and $a$.
For small $a>0$ and $r\rightarrow 0$, the dynamic tipping window limits to the chaotic tipping window. For larger $r$ it shifts to larger values of $\eta$. For $r, a>0$ fixed and $a$ small enough, it is the smallest range of $\eta$ values for which the nonautonomous system \eqref{eq:fullSystAdd} undergoes tipping (i.e., the $x$ dynamics pass the tipping threshold (here $\Psi=0.1$)) in response to chaotic forcing typical with respect to some measure $s\in \cS(m)$. Note that the dynamic tipping window depends on the choice of the tipping threshold.

\subsection{Example: a ramped Lorenz-Stommel System}
\label{sec:DynamicTippingWindowLS}

We consider the two-dimensional Lorenz-Stommel system (\ref{eq:Lorenz-Stommel}), replacing the parameter $\eta_1$ with a time-dependent parameter $\eta_{r}(t)$ given by (\ref{eq:eta_of_t}) with $\eta_0=0.9$.
We initialize the $x$-dynamics in the vicinity of the AMOC-on state at $(x_1,x_2)=(1.7,\, 1.0)$, and then run ensembles for three different values of $r$, a range of values of $a\in [0.0,0.08]$, and the forcing always initialized by the final condition of the previous run. Thus, we approximate random forcing trajectories.
We run the system until time $t_c$ when $\Psi$ crosses the tipping threshold $\Psi=0.1$ and store $\eta_r(t_c)$ for each run.
The results are shown in Figure~\ref{fig:LSdynTipWin}. As expected, we observe that for increasing $r$, the dynamic tipping window shifts to larger $\eta_1$ for all $a$. 

\begin{figure}
    \centering
    
    \includegraphics[align=t,width=10cm]{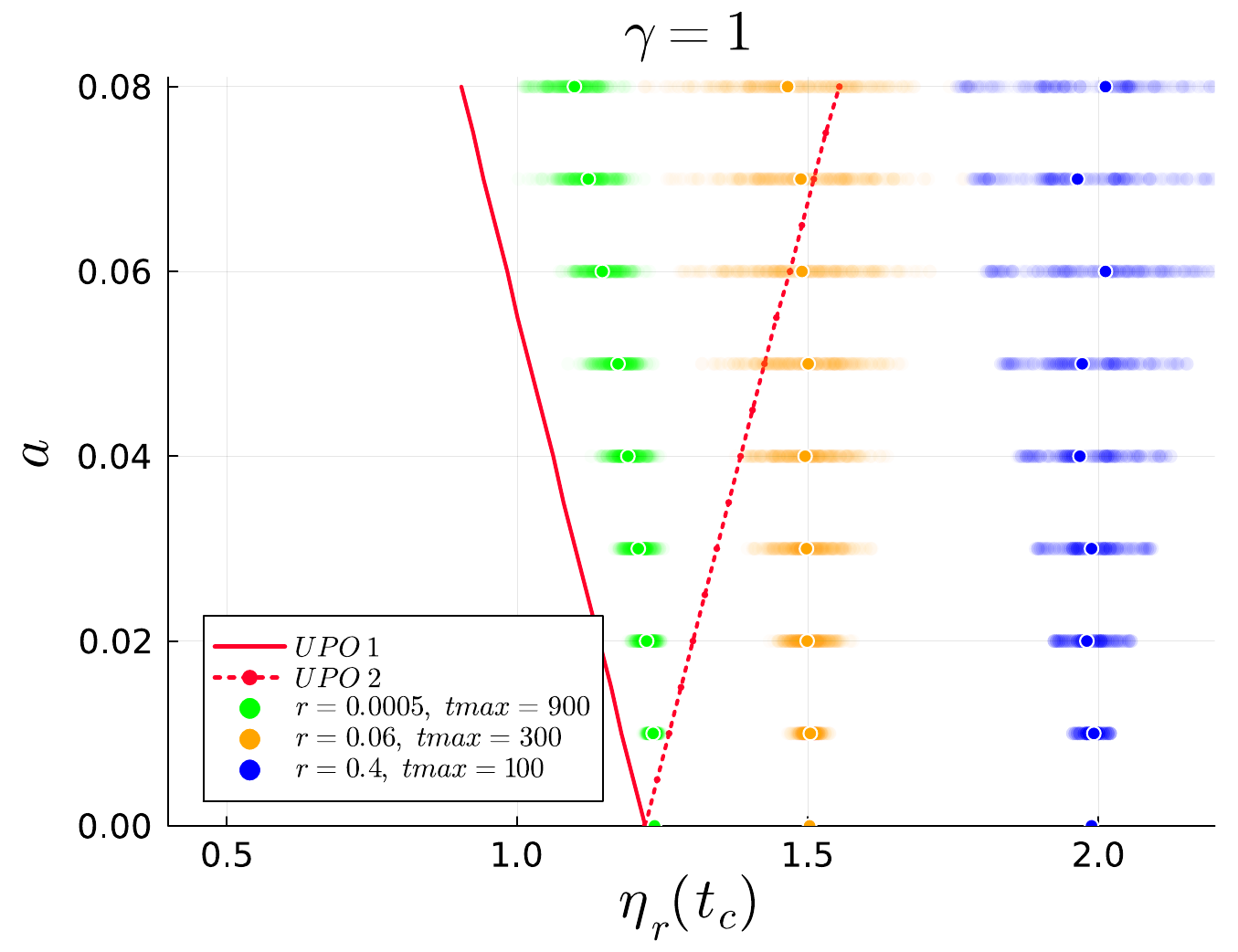}  

    \caption{The dynamic tipping window: we initialize the $x$-dynamics in the basin of the AMOC-on state at $(x_1,x_2)=(1.7,\, 1.0)$, and then do ensembles of runs for three different values of $r$, a range of values of $a\in [0.0,0.08]$, and the forcing always initialized by the final condition of the previous run. Thus, we approximate random forcing trajectories.
    We run the system until time $t_c$ when the system crosses the tipping threshold $\Psi=0.1$ and store $\eta_r(t_c)$ for each run. For each run, we plot the location of $\eta_r(t_c)$, and for each line of constant $a$, we show the median tipping location for all three choices of $r$.
    The points are plotted slightly transparent such that one can see the distribution of $\eta_r(t_c)$.}
    \label{fig:LSdynTipWin}
\end{figure}

\subsection{Phase Dependence of UPO forcing for time-dependent $\eta_1(t)$}

Instead of forcing the Stommel system with chaotic trajectories, we could have forced it with the UPOs. In this case the value of $\eta_r(t_c)$ would depend on the initial condition chosen on each UPO. Thus, we would get a range of $\eta_r(t_c)$ for each UPO and forcing strength $a$. However, in the limits of infinite timescale separation, these ranges should converge to single values of $\eta_r(t_c)$.

This phenomenon is not to be confused with the {\it phase tipping} introduced in \cite{alkhayuon2021phase}, where a tipping event happens because of a combination of the forcing shifting at a certain rate and the system being in a certain phase with respect to its limit cycle. In our setting, one of the system's attractors will lose stability anyway for sufficiently large $t$ as the chaotic forcing is bounded and $a$ is finite, and even an unusually stabilizing forcing trajectory could not keep both attractors of the $x$-dynamics stable for arbitrarily large $\eta_r(t)$. Thus, the forcing only determines when the system tips, and the (de)stabilization of the system should be rather interpreted as an advancing or delaying of the bifurcation.

\section{Discussion}
\label{sec:discuss}

In this work, we provide an understanding of the chaotic tipping window for chaotically forced ODEs that can exhibit bistability for certain parameter values. This extends \cite{AshwinNewmanRoemer:2024} and demonstrates that in the forced ODE system the timescale separation constant $\gamma$ is vital to understand the relative importance of the mean (for rapid forcing) or extremes (for slow forcing) of the forcing, not just its amplitude. We show that the timescale ratio between the forcing and responding systems strongly affects the forcing behavior that determines the boundaries of the tipping window; for forcing at similar timescales to the bistable system we expect quite complex behavior is possible.

The tipping window encodes information about UPOs on the chaotic forcing attractor and about the response system. For example, most of the considered UPOs on the Lorenz attractor in Figure~\ref{fig:UPOs} have similar extrema while their means can vary significantly, which is reflected in the different sizes and structures of the tipping window in the different limits of timescale separation. 
We recall that finding the endpoints of the tipping window can be seen as a non-linear version of ergodic optimisation that may be interesting to understand at a more systematic level.
Some work has been done to understand bifurcations with bounded noise in terms of set-valued dynamics \cite{lamb2015topological,kuehn2018early}, but it is not clear that this will more easily give information about the tipping windows.
We note that we do not give results of relevance to the homogenization limit \cite{Kelly:2017, Deser2025chaotic} as this limit requires that $a\approx \sqrt{\gamma}$ and takes the limit $\gamma\rightarrow\infty$; the large forcing amplitude means that the tipping window becomes trivial in this limit.

In the presence of bounded noise and slow drift of a parameter, we show that, not surprisingly, the tipping window becomes dynamic with a delay related to the parameter's drift rate. We did not consider the interplay between the rate of drift and the timescales of forcing and response, but presumably, these will all be important, especially when they are comparable. We also did not consider, for example, the cumulative distribution of tipping times for a drift through a chaotic tipping window, except that some features of this can be inferred from Figure \ref{fig:DWgamma_intermediate} - the shading gives an idea of the time of tipping when parameters reside in these regions.
It will be an interesting and potentially useful exercise to quantify this more. As in \cite{AshwinNewmanRoemer:2024}, we expect that an extremely slow drift will be needed for tipping to be likely close to the left boundary of the chaotic tipping window.
For more complex real-world cases where there is a combination of chaotic and slow deterministic forcing there may be a complex interplay of the forcing and response.

\section*{Acknowledgements}

PA received funding from the European Union’s Horizon 2020 research and innovation programme under grant agreement No. 101137601 (ClimTip). RR received funding from the European Union’s Horizon 2020 Research and Innovation Programme under the Marie Skłodowska-Curie Grant Agreement No. 956170 (CriticalEarth).
Note that for the purpose of open access, we have applied a Creative Commons Attribution (CC BY) license to any Author Accepted Manuscript version arising from this submission.

\section*{Data availability}

The Julia code to perform the simulations and generate the figures in this paper is available from 
{\tt https://github.com/raphael-roemer/tipping-windows-and-timescales-2025}.

\bibliographystyle{plain}
\bibliography{nonautrefs}

\clearpage
\newpage

\appendix

\section{Computation of unstable periodic orbits}
\label{sec:PO_computation}

We compute the UPOs using the method from \cite{davidchack1999efficient, Eckmann_1987, Auerbach} as follows. Consider the recurrence matrix $M_{rec}$ of a long trajectory of the Lorenz-63 system defined by:
$$
	M_{rec}(ij) = \begin{cases}
				1, & \text{if } |{\bf y}(t_i)- {\bf y}(t_j)| \leq \varepsilon \\
            			0, & \text{otherwise.}
			\end{cases}
$$
An example is shown in Figure \ref{fig:M_rec}. 
Then, from each stripe in the recurrence matrix, we choose the candidate $M_{rec}(ij)$ with the smallest recurrence distance $|{\bf y}(t_i)- {\bf y}(t_j)|$. This provides a list of approximated periodic orbits.

Many of the identified UPOs can now either be directly identified with each other, or they are identical after a rotation (i.e. they are identical if one applies the transformation $(y_1,y_2,y_3) \mapsto (-y_1,-y_2,y_3)$) or they are similar in the sense that a UPO is run through several times and thus has a length that is an integer multiple of its actual length. Thus, we identify many of the previously collected candidates for periodic orbits with each other.
We illustrate a number of the resulting UPOs computed for the Lorenz attractor in Figure~\ref{fig:UPOs} and list some properties of these UPOs in Table~\ref{tab:POs}.

\begin{figure}
    \centering
    \includegraphics[align=t,width=12cm]{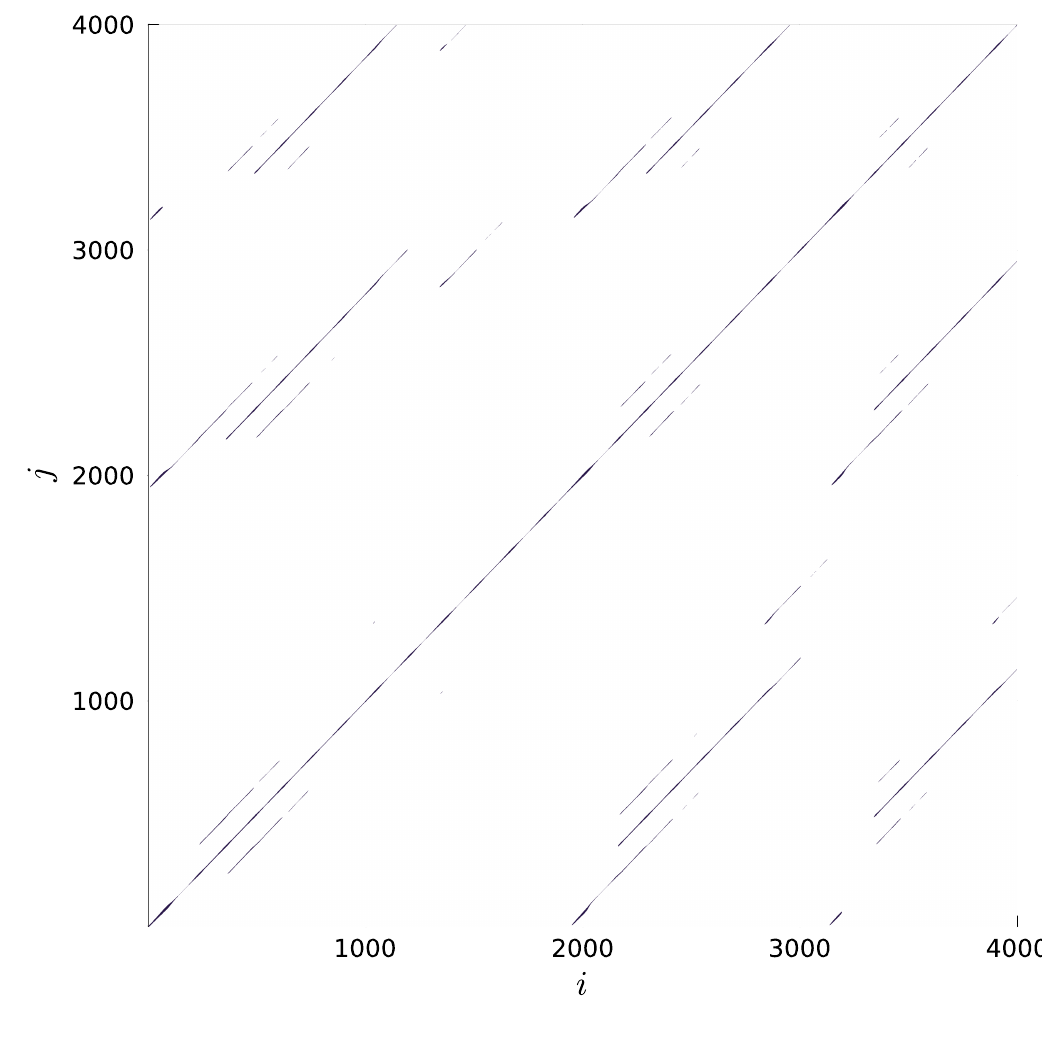}  

    \caption{The first 4000 timesteps of the recurrence matrix $M_{rec}(i,j)$ for a numerically integrated trajectory of the Lorenz system (\ref{eq:Lorenzz}) starting at $(y_1,y_2,y_3)=(0.05,0.05,25.05)$. The $i$ and $j$ axes refer to timestep number $t_i$ and $t_j$. Black indicates a recurrence where $|{\bf y}(t_i)-{\bf y}(t_j)|\leq \varepsilon = 0.9$. In each of the black stripes, we search for a pair $(i,e)$ with closest recurrence and approximate orbit closure and use this as the candidate for a point on the UPO.}
    \label{fig:M_rec}
\end{figure}

\begin{figure}
    \centering

    \includegraphics[align=t,width=13cm]{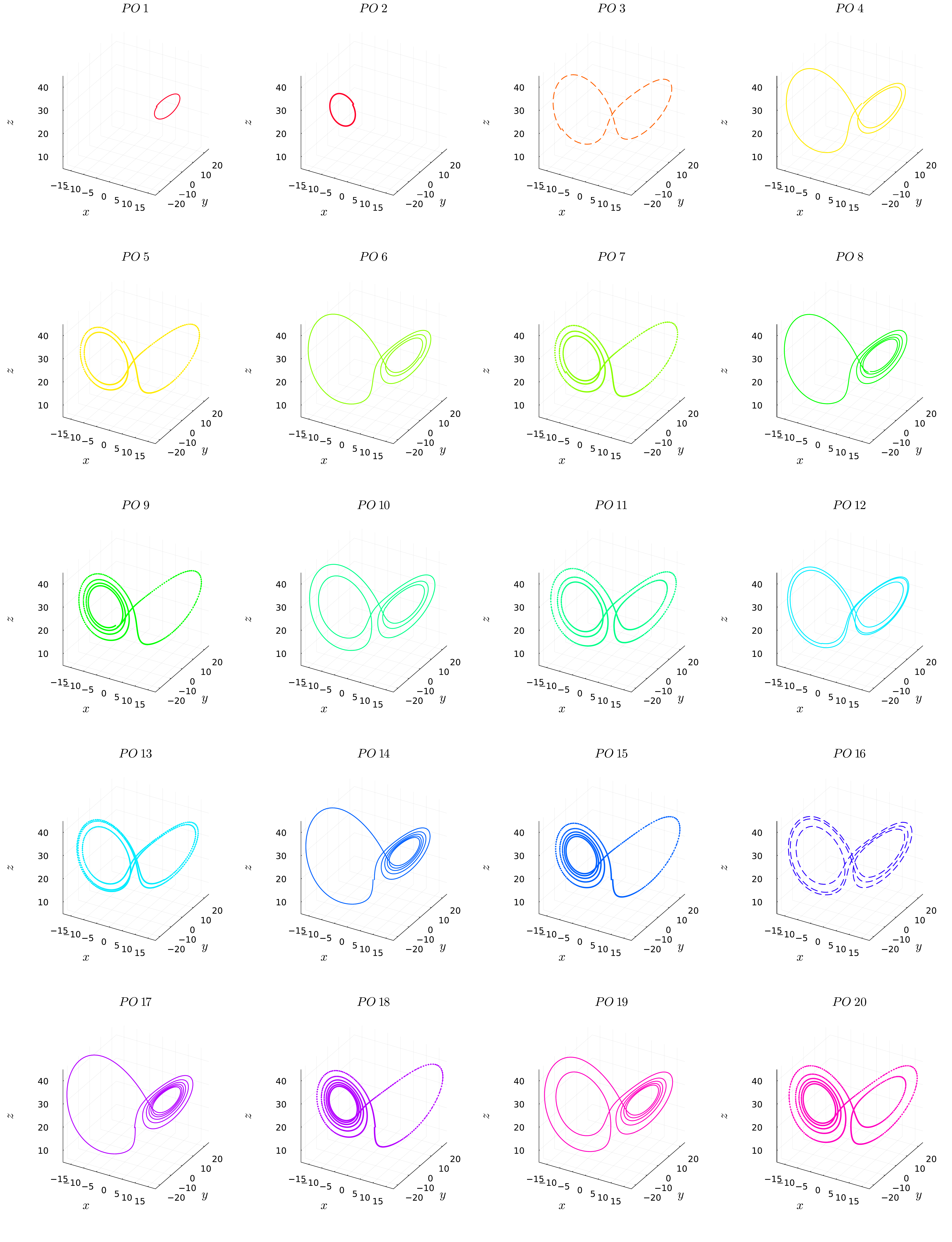}  

    \caption{Numerical approximations of a set of low period UPOs of the Lorenz system within the Lorenz attractor, computed using the method of recurrences from a long chaotic trajectory on the Lorenz attractor. Dashed lines show axis-symmetric orbits. Orbits with the same color are symmetrically related.}
    \label{fig:UPOs}
\end{figure}

\begin{table}
    \centering
    \begin{tabular}{c|c|c|c|c}
        Name  &  $T$   & $y_{1,max}$	& $y_{1,mean}$ & symmetric \\ \hline
        UPO1  	&  0.640  	&  11.751	&  8.050     &   	 \\
        UPO2  	&  0.640   	&  -5.135  	&  -8.050    &   	 \\
        UPO3 	&  1.545  	&  15.661  	&  -0.073     & Y  \\
        UPO4 	&  2.285  	&  14.601  	&  2.301  	 &  	 \\
        UPO5  	&  2.285  	&  16.532	&  -2.301  	 &   	 \\
        UPO6  	&  3.020 	&  14.947  	&  3.456 	 &   	 \\
        UPO7 	&  3.020  	&  17.047  	&  -3.456	 &   	 \\
        UPO8  	&  3.700  	&  14.957	&  4.177  	 &       \\
        UPO9  	&  3.700  	&  17.000 	&  -4.177 	 &   	 \\
        UPO10 	&  3.825  	&  15.945  	&  1.361	 &   	 \\
        UPO11  &  3.825  	&  16.719	&  -1.361  	 &   	 \\
        UPO12	&  3.865  	&  15.430  	&  1.325 	 &   	 \\
        UPO13 	&  3.865  	&  16.231  	&  -1.325	 &       \\
        UPO14 	&  4.410  	&  14.757  	&  4.701  	 &  	 \\
        UPO15  &  4.410  	&  17.602	&  -4.701  	 &   	 \\
        UPO16  &  4.635  	&  16.188  	&  -0.001 	 & Y  \\
        UPO17 	&  5.110  	&  14.957  	&  5.051	 &       \\
        UPO18  &  5.110  	&  17.828	&  -5.051  	 &       \\
        UPO19  &  5.235  	&  15.864  	&  2.970 	 &   	 \\
        UPO20 	&  5.235  	&  17.440  	&  -2.970	 &  
    \end{tabular} 
    \caption{Approximate periods $T$, maximum $y_{1,max}$ and  mean $y_{1,mean}$ of the first component of the numerically determined UPOs $({\bf y}(t))$ shown in Figure \ref{fig:UPOs}. Note that the symmetric orbits have a mean that is approximately $0$.}
    \label{tab:POs}
\end{table}

\end{document}